\documentclass[10pt]{IEEEtran}

\ifCLASSOPTIONcompsoc
\usepackage[nocompress]{cite}
\else
\usepackage{cite}
\fi
\usepackage{tikz}


\usepackage{algorithmic}
\usepackage{amsthm}
\usepackage{amsmath,amsfonts}
\usepackage{subcaption}
\usepackage{bm}
\usepackage{bbding}
\usepackage{array} 
\usepackage{listings}
\usepackage{booktabs}
\usepackage{xcolor}
\usepackage{tabularx}
\newtheorem{theorem}{Theorem}
\newtheorem{lemma}{Lemma}

\theoremstyle{definition}

\newcolumntype{C}{>{$}c<{$}}

\newcommand{\GGG}  {\mathcal{G}}

\newcommand{\ZZ}  {\mathbb{Z}}

\newcommand{\Pdot}  {\Pi_\mathsf{Inner}}

\newcommand{\Pvdot}  {\Pi_\mathsf{InnerVerify}}
\newcommand{\Psha}  {\Pi_{[\cdot]}}
\newcommand{\Pshc}  {\Pi_{\langle\cdot\rangle}}

\newcommand{\Pdabits}  {\Pi_{\mathsf{edaBits}}}

\newcommand{\Pmulv}  {\Pi_\mathsf{MultVerify}}

\newcommand{\Prd}  {\Pi_\mathsf{Reduce}}
\newcommand{\Pple}  {\Pi_\mathsf{PolyEvl}}
\newcommand{\Ptran}  {\Pi_\mathsf{Trans}}

\newcommand{\Pmul}  {\Pi_\mathsf{Mult}}

\newcommand{\Ptrunc}  {\Pi_\mathsf{Trunc}}

\newcommand{\Prec}  {\Pi_\mathsf{Rec}}

\newcommand{\Pbsv}  {\Pi_\mathsf{BIVerify}}

\newcommand{\myhalfbox}[5]{
	\begin{figure}[tpb]
		\centering
		\begin{tikzpicture}
		\node[anchor=text,text width=\columnwidth-1.1cm, draw, line width=1pt, fill=white!5, inner sep=5mm, font=\fontsize{8}{10}\selectfont] (big) {\\#4};
		\node[draw,  line width=.5pt, fill=white!10, anchor=west, xshift=5mm] (small) at (big.north west) {#1};
		\end{tikzpicture}
		\caption{#5}
		\vspace{-1em}
	\end{figure}
}

\usepackage{enumitem}
\usepackage{amsthm}

\usepackage{multirow}
\usepackage{color}
\usepackage{pifont}
\usepackage{graphicx}
\usepackage{listings}
\usepackage{listings} 
\usepackage{xcolor}
\usepackage{float}
\usepackage{mathtools}
\usepackage{caption}
\usepackage{stmaryrd}

\begin{document}

\title{The Communication-Friendly Privacy-Preserving Machine Learning against Malicious Adversaries}

\author{Tianpei~Lu,
	Bingsheng~Zhang,~\IEEEmembership{Member,~IEEE,}
	Lichun~Li,
	and~Kui~Ren,~\IEEEmembership{Fellow,~IEEE}
	\thanks{T.~Lu, B.~Zhang and K.~Ren are with The State Key Laboratory of Blockchain and Data Security, Zhejiang University, Zhejiang University, Hangzhou, China.~E-mail:~\{lutianpei, bingsheng, kuiren\}@zju.edu.cn. L.~Li is with Ant group, Hangzhou, China.~E-mali:~lichun.llc@antgroup.com. B.~Zhang is Corresponding Author.}
}
\maketitle
\begin{abstract}
With the increasing emphasis on privacy regulations, such as GDPR, protecting individual privacy and ensuring compliance have become critical concerns for both individuals and organizations. Privacy-preserving machine learning (PPML) is an innovative approach that allows for secure data analysis while safeguarding sensitive information. It enables organizations to extract valuable insights from data without compromising privacy. Secure multi-party computation (MPC) is a key tool in PPML, as it allows multiple parties to jointly compute functions without revealing their private inputs, making it essential in multi-server environments. We address the performance overhead of existing maliciously secure protocols, particularly in finite rings like $\mathbb{Z}_{2^\ell}$, by introducing an efficient protocol for secure linear function evaluation. We implement our maliciously secure MPC protocol on GPUs, significantly improving its efficiency and scalability. We extend the protocol to handle linear and non-linear layers, ensuring compatibility with a wide range of machine-learning models. Finally, we comprehensively evaluate machine learning models by integrating our protocol into the workflow, enabling secure and efficient inference across simple and complex models, such as convolutional neural networks (CNNs).
\end{abstract}


\section{Introduction}
In the era of big data, privacy protection and compliance have become paramount concerns for both individuals and organizations. As various privacy regulations, such as GDPR, have emerged, the demand for effective privacy-preserving mechanisms has intensified significantly. Privacy-preserving machine learning (PPML) is an innovative technique that enhances privacy while enabling secure data mining and machine learning. It ensures that sensitive information remains confidential, allowing organizations to leverage data insights without compromising individual privacy.

Secure multi-party computation (MPC)~\cite{mpc1,mpc2,mpc3} allows multiple parties to jointly evaluate functions without revealing their private inputs. This cryptographic tool plays a crucial role in realizing PPML in multi-server environments~\cite{aby3, BLAZE, FLASH, Chameleon, epic, FALCON}. Notably, this work focuses on 3-party MPC, referred to as 3-PC. Most existing protocols~\cite{SecureNN, ASTRA} are designed for a semi-honest setting, where participants are assumed to adhere to the protocol and act honestly, albeit with the potential to glean additional information from the data they handle. However, in many scenarios, the importance of robust defenses against malicious actors becomes critical. Maliciously secure protocols are essential in these contexts, as they can detect adversarial behaviors and protect the integrity of the computation.

Despite the advancements, state-of-the-art maliciously secure PPML protocols face significant performance overhead. For instance, maliciously secure multiplication protocols can be at least twice as slow as their semi-honest counterparts~\cite{swift, FastRing}. This performance gap raises concerns, especially given that PPML-friendly MPC protocols typically operate over finite rings like $\ZZ_{2^\ell}$, which facilitate fixed-point arithmetic. Designing maliciously secure MPC over $\ZZ_{2^\ell}$ is inherently more complex than over prime-order finite fields $\ZZ_p$.

Recently, several works~\cite{GS20, TurboPack, nv18} have successfully implemented efficient maliciously secure protocols over $\ZZ_p$. However, techniques used to achieve malicious security in $\ZZ_p$ cannot be directly applied to $\ZZ_{2^\ell}$ due to the absence of inverses for certain elements. Attempts to adapt these techniques have resulted in protocols that incur a twofold communication overhead. Alternatively, some research efforts~\cite{BLAZE, swift, FastRing} aim to develop maliciously secure MPC over $\ZZ_{2^\ell}$ from the ground up. Nonetheless, these solutions often generate significantly higher communication overhead compared to semi-honest protocols. This performance loss is particularly troubling in today’s economic landscape, where communication costs on platforms like Amazon can far surpass computation costs, underscoring the urgent need for efficient, secure protocols that balance both privacy and performance.

\smallskip
\noindent\textbf{Our results.}
In this work, we improve the performance of maliciously secure linear functions evaluation for enhanced PPML.  Our protocols are based on 3-party MPC in the honest majority setting. The underlying share of our 3-PC protocol originates from a variant of the replicated secure sharing (RSS)~\cite{ASTRA}; that is, to share $x\in\ZZ_{2^\ell}$, $P_0$ holds $(r_1,r_2)$, $P_1$ holds $(m = x - r, r_1)$, and $P_2$ holds $(m = x- r, r_2)$ where $r=r_1+r_2$. 

Analogously, for the malicious multiplication, the parties first invoke the semi-honest multiplication protocol and perform a batch verification at the end. Goyal {\em et al.}~\cite{GS20} proposes a technique that can transfer the verification of $N$ dimension inner product triple to the verification of $N/2$ dimension inner product with constant overhead. However, Goyal {\em et al.}~\cite{GS20} works on Shamir's secret sharing, which is performed over a prime-order field, naively converting their protocol to the ring setting could cause the soundness issue. Also, as mentioned above, the techniques~\cite{Yet, spdz2k, Brain} to adopt the multiplication verification over the field to the ring are not suitable for the protocol proposed in \cite{GS20}. To resolve the soundness issue, we extend the shared elements over $\ZZ_{2^\ell}$ to the quotient ring of polynomials $\ZZ_{2^\ell}[x]/f(x)$~\cite{BBC,BGIN,BGIN2}, where $f(x)$ is a degree-$d$ irreducible polynomial over $\ZZ_{2^\ell}$ to apply the Lagrange interpolating based dimension reduction technique~\cite{GS20}. Consequently, the overall communication of our batch multiplication verification protocol is logarithmic to the number of multiplication gates.

Our protocols are compatible with mixed-circuit computation. Previous research~\cite{TASTY, SecureML, aby3, Trident} has shown that computing non-linear functions, such as comparison, is more efficient in binary computation. This necessitates switching between arithmetic and binary computation, as arithmetic is superior for dot products. Rotaru and Wood introduced the concept of double-authenticated bits (daBits)~\cite{dabits}, which are secret random bits shared across both arithmetic and binary. We observe that our protocol can be directly applied to daBits with minimal modifications. By utilizing daBits, we enable secure evaluation of any non-linear function under malicious security.

Finally, we integrated both linear and non-linear functions to systematically evaluate machine learning models.

\begin{table*}[tbh!]
\centering
\caption{Comparison of 3-PC based PPML. ($\ell$ is the  ring size, $n$ is the size of the inner product.)}

\newcolumntype{Y}{>{\raggedleft\arraybackslash}X}
\begin{tabularx}{\textwidth}{X|X|XXX|X}
\toprule
\multirow{2}{*}{Operation}      & \multirow{2}{*}{Protocol} & Offline       & \multicolumn{2}{c|}{Online} & \multirow{2}{*}{Malicious} \\ \cmidrule(r){4-5} 
                                &                       & Communication (bits) & Rounds     & Communication (bits)     &                           \\ \midrule
\multirow{4}{*}{Mult} &  ABY3\cite{aby3}                & $12\ell$        &     $1$       &    $9\ell$   &   $ \checkmark $     \\
				    & BLAZE\cite{BLAZE}        & $3\ell$        &     $1$       &    $3\ell$   &    $\checkmark$      \\
                                &      SWIFT\cite{swift}     &   $3\ell$  &   $1$    &    $3\ell$     &      $\checkmark$             \\
                                &        \textbf{Ours}      &   $1\ell$   &  $1$    & $2\ell$     &      $\checkmark$       \\ \midrule
\multirow{4}{*}{Inner Product} & ABY3\cite{aby3}                  & $12n\ell $        &     $1$       &    $9n\ell$   &    $\checkmark$      \\ 
				& BLAZE\cite{BLAZE}                 & $3 n \ell $        &     $1$       &    $3\ell$   &    $\checkmark$      \\
                                &    SWIFT\cite{swift}      &   $3 \ell$  &   $1$    &    $3\ell$     &      $\checkmark$             \\
                                &        \textbf{Ours}      &   $1 \ell$   &  $1$    & $2\ell$     &      $\checkmark$       \\ \midrule
\multirow{4}{*}{\begin{tabular}[l]{@{}l@{}}Inner Product\\ with \\ Trunction\end{tabular} } & ABY3\cite{aby3}                   & $12 n\ell +84\ell$        &     $1$       &    $9 n\ell +3\ell$   &    $\checkmark$      \\ 
				& BLAZE\cite{BLAZE}                  & $3 n\ell + 2\ell $        &     $1$       &    $3\ell$   &    $\checkmark$      \\
                                &      SWIFT\cite{swift}      &   $15\ell$  &   $1$    &    $3\ell$     &      $\checkmark $            \\
                                &        \textbf{Ours}      &   $7\ell$   &  $1$    & $2\ell$     &      $\checkmark $      \\ \bottomrule
\end{tabularx}
\label{tab:cmppr}
\vspace{-1em}
\end{table*}

\smallskip
\noindent\textbf{Performance.} Table~\ref{tab:cmppr} depicts the  comparison between our protocols and  SOTA 3PC maliciously secure protocol. As we can see, our protocols achieve a significant communication reduction.

 \emph{\underline{Batch verification for multiplication over the ring.}} Compared with the prime-order finite field, constructing an MPC over ring $\ZZ_{2^\ell}$ against malicious adversaries typically incurs a higher overhead. In this work, we propose a new maliciously secure 3PC multiplication protocol over ring $\ZZ_{2^\ell}$ with a logarithmic communication overhead during batch verification. 
We conduct benchmarks on the overhead ratio of the verification step. 
By employing this technique,  the amortized communication cost of our maliciously secure multiplication is merely $2$ ring elements in the online phase and $1$ ring element in the offline phase per operation. 

Compared with SOTA  maliciously secure MPC multiplication over ring proposed by Dalskov {\em et al.}~\cite{FastRing}, our protocol reduces the overall communication by 40\%. Note that Dalskov {\em et al.}~\cite{FastRing} achieves full security in the $\mathcal{Q}^3$ active adversary setting ($t<n/3$), while our protocol achieves security with abort in the $\mathcal{Q}^2$ active adversary setting ($t<n/2$),  where $t$ is the number of corrupted parties and $n$ is the total number of participants.
Compared with SOTA 3PC multiplication over ring~\cite{swift}, our protocol reduces the communication by 33\% in the online phase and 67\% in the offline phase, respectively. Similarly, the communication of our inner product protocols is also 50\% of that in SWIFT~\cite{swift}.

\emph{\underline{Implementation with GPUs.}} Since our implementation requires converting secret sharing to an extended ring during the verification phase, this introduces significant computational overhead. However, the extended ring offers excellent concurrency, allowing us to implement our protocol on GPUs. In our specific experiments, compared to ABY3, our implementation achieved a threefold performance improvement, and when compared to Swift, we realized a twofold increase in performance.

\emph{\underline{Implementation of maliciously secure PPML framework.}} We built a comprehensive privacy-preserving machine learning application against malicious adversaries based on Piranha~\cite{piranha} framework. This includes the implementation of typical CNN models such as VGG and ResNet. Our framework delineates between semi-honest offline and online computation phases, as well as a separate multiplication gate (for both arithmetic and boolean) verification phase. Our experiments demonstrate that the time overhead of the verification phase is significantly lower than that of the online computation phase, indicating that the time introduced by malicious security is far less than the original cost of the semi-honest protocol.

\noindent\textbf{Paper Organization.}  We first propose our maliciously secure 3PC in Sec.~\ref{sec:tomal}. In Sec.~\ref{apparithmetic}, we realize the PPML framework based on our maliciously secure protocols for both linear and non-linear operation. In Sec.~\ref{imp_ben}, we benchmark the performance of our protocols and PPML framework.


\section{Preliminaries}


\noindent\textbf{Notation.} Let $\mathcal{P}:=\{P_0,P_1,P_2\}$ be the three MPC parties.
 During the PPML execution, we encode the float numbers as fixed-point structure~\cite{aby3, BLAZE}: for a fixed point value $x$ with $k$-bit precision, if $x\geq 0$, we encode it as $\lfloor x\cdot 2^k \rfloor$; if $x<0$, we encode it as $2^\ell + \lfloor x\cdot 2^k \rfloor$. 
 We use $\eta_{j,k}$ to denote the common seed held by $P_j$ and $P_k$. Our protocol contains two types of secret sharing as follows:
\begin{itemize}
\item $[\cdot]^\ell$-sharing: We define $[\cdot]^\ell$-sharing over ring $\ZZ_{2^\ell}$ as  $[x]^\ell := ([x]_1\in \ZZ_{2^\ell}, [x]_2\in \ZZ_{2^\ell})$ where $x = [x]^\ell_1 + [x]^\ell_2$. $P_j$ for $j\in\{1,2\}$ hold share $[x]^\ell_j$.
\smallskip
\item $\langle\cdot\rangle^\ell$-sharing: We define $\langle \cdot \rangle^\ell$-sharing over ring $\ZZ_{2^\ell}$ as  $\langle x\rangle^\ell := ([r_x]^\ell, m_x)$ where $r_x$ is a fresh random value and $m_x = r_x + x$. $P_j$ for $j\in\{1,2\}$ hold $(m_x\in \ZZ_{2^\ell}, [r_x]^\ell_j \in \ZZ_{2^\ell} )$ and $P_0$ holds  $([r_x]^\ell_1, [r_x]^\ell_2 )$.

\end{itemize}

 We use $[\cdot]^{\ell[x]}$ and $\langle\cdot\rangle^{\ell[x]}$ to denote the share in the polynomial ring $\ZZ_{2^\ell}[x]/f(x)$ where $f(x)$ is a degree-$d$ irreducible polynomial over $\ZZ_2$. 
 For simplicity, we use $[\cdot]$, $\langle \cdot \rangle$ when semantics are clear.

All the aforementioned secret-sharing forms have the linear homomorphic property, i.e., $[x]+[y] = ([x]_1 + [y]_1, [x]_2 + [y]_2)$ and $c\cdot[x] = (c\cdot [x]_1, c\cdot [x]_2)$ and $[x]+c = ([x]_1 + c, [x]_2)$, where $c$ is a public value. The same linear operation holds for $\langle \cdot \rangle$, and $\langle\cdot\rangle^{\ZZ_{2^\ell}[x]}$.

\smallskip
\noindent\textbf{Secret sharing.} 
Let $\Psha$ and $\Pshc$ denote the corresponding secret-sharing protocols.  By $\Psha(x)$, we mean that $x$ is shared by $P_0$; by $\Psha$, we mean the parties jointly generate a shared random value.
We utilize pseudo-random generators (PRG) to reduce the communication~\cite{YAS}. 
In our protocol description, when we let parties $P_j$ and $P_k$ pick random values together, we mean that these parties invoke PRG with seed $\eta_{j,k}$. 
The brief sketch of secret sharing schemes is as follows.

\begin{itemize}
\item $[x]^\ell \leftarrow \Psha^{\ell}(x)$: (Generate shares of $x$.)

\quad - $P_0$ and $P_1$ pick random value $[x]_1\in \ZZ_{2^\ell}$ with seed $\eta_{0,1}$; 

\quad - $P_0$ sends $x_2 = x - [x]_1 \pmod {2^\ell}$ to $P_2$.

\item  $[x]^\ell \leftarrow \Psha^{\ell}$:  (Generate shares of a random value.)

\quad - $P_0$ and $P_1$ pick random value $[x]_1\in \ZZ_{2^\ell}$ with seed $\eta_{0,1}$; 

\quad - $P_0$ and $P_2$ pick random value $[x]_2\in \ZZ_{2^\ell}$ with seed $\eta_{0,2}$; 

\quad - $P_0$ calculates $x = [x]_1 + [x]_2$.

\item $\langle x \rangle^\ell \leftarrow \Pshc^{\ell, k}(x)$:  (Generate shares of $x$.)

\quad - All parties perform $[r_x] \leftarrow \Psha$ in the offline phase, and $P_k$ holds both seeds of $[r_x]_1$ and $[r_x]_2$ generation; 

\quad - $P_i$ send $m_x = x + [r_x]_1 + [r_x]_2$ to $P_1$ and $P_2$. 

\item $\langle x \rangle^\ell \leftarrow \Pshc^\ell$:   (Generate shares of a random value.)

\quad - All parties perform $[r_x] \leftarrow \Psha$ in the offline phase; 

\quad - $P_1$ and $P_2$ pick random value $m_x$ with seed $\eta_{1,2}$.  

\end{itemize}
$\Psha$ and $\Pshc$ also work for the share $[\cdot]^{\ell[x]},\langle\cdot\rangle^{\ell[x]}$ over the polynomial ring $\ZZ_{2^\ell}[x]/f(x)$, which are denoted as $\Psha^{\ell[x]}$, $\Pshc^{\ell[x]}$. 
%
%
%
%

\smallskip
\noindent\textbf{Verifiability of share reconstruction.} We note that the shared form $\langle \cdot \rangle$ has the verifiable reconstruction property against a single malicious party. To be precise, for shared value, $\langle x \rangle$, a single active adversary cannot deceive the honest parties into accepting an incorrect reconstruction result $x + e$ with a non-zero error $e$. This is because any two honest parties can collaboratively reconstruct the secret, and invalid shares will be detected by the honest parties.

Formally, the verifiable reconstruct protocol $\Prec$ is described as follows:
\begin{itemize}
\item $x \leftarrow \Prec(\langle x \rangle)$: 

\quad - $P_0$ sends $[r_x]_1$ to $P_2$ and $[r_x]_2$ to $P_1$;

\quad - $P_1$ sends $m_x$ to $P_0$ and $H([r_x]_1)$ to $P_2$; 

\quad - $P_2$ sends $H(m_x)$ to $P_0$ and $H([r_x]_2)$ to $P_1$; 

If the received messages from the other parties are inconsistent, $P_i$ output abort. Otherwise $P_i$ output $x = m_x - [r_x]_1-[r_x]_2$. 
\item $x \leftarrow \Prec^{\ell, k}(\langle x \rangle)$: All parties send their shares (or the hash value) to $P_k$. If the received messages from the other parties are inconsistent, $P_k$ output abort. Otherwise $P_k$ output $x = m_x - [r_x]_1-[r_x]_2$.
\end{itemize}
For the share $\langle\cdot\rangle^{\ell[x]}$ in polynomial ring, $\Prec^{\ell[x]}$ works analogously as the above. 

\smallskip
\noindent\textbf{Preprocessing and postprocessing.}
We follow the ``preprocessing" paradigm \cite{BDOZ}, which splits the protocol into two phases: the preprocessing/offline phase is data-independent and can be executed without data input, and the online phase is data-dependent and is executed after data input.  Specifically, all the items $r_x$ of share $\langle x \rangle$ of our protocols can be generated in the circuit-depend offline phase. What the parties need to do in the online phase is to collaborate in computing $m_x$ for $P_1$ and $P_2$. To achieve malicious security, we further introduce the postprocessing phase \cite{Brain}, where batch verification is performed.

\smallskip
\noindent\textbf{Multiplication gate.}
We adopt the multiplication protocol of ASTRA\cite{ASTRA}. 
For multiplication $z = x\cdot y$ with input $\langle x\rangle$, $\langle y\rangle$ and output $\langle z \rangle$, all parties first generate $[r_z] \leftarrow \Psha(r_z)$ for the output wire in the offline phase. To calculate $m_z$ for $P_1$ and $P_2$ in the online phase, it can be written as
\begin{equation}
\begin{aligned}
	m_z = xy + r_z &= (m_x - r_x)(m_y - r_y) + r_z \\
	&=\overbrace{m_x m_y - m_x r_y - m_y r_x}^{P_1 \text{ and } P_2 \text{ can locally evaluate} } + \overbrace{r_x r_y  +r_z}^{\text{Known to } P_0}   \enspace . \nonumber
\end{aligned}
\end{equation}
$[\Gamma'] = m_x m_y - m_x [r_y] - m_y [r_x]$ can be calculated by $P_1$ and $P_2$ locally and $[\Gamma] = [r_x \cdot r_y]  -[r_z]$ can be secret shared by $P_0$ to $P_1$ and $P_2$ in the preprocessing phase. In the online phase, $P_1$ and $P_2$ calculate and reconstruct $[m_z] = [\Gamma'] + [\Gamma]$.

\noindent\textbf{Inner product.} Given an arbitrary dimension inner product, its communication cost equals to a single multiplication. Considering $n$-dimension inner product $z = \sum^{n-1}_{i = 0} x_i \cdot y_i$, the artifact $m_z$ requires to be evaluated in online phase can be written as

\begin{equation}
\begin{aligned}
	m_z &= \sum^{n-1}_{i = 0} x_i \cdot y_i + r_z = \sum^{n-1}_{i = 0}(m_{x_i} - r_{x_i})(m_{y_i} - r_{y_i}) + r_z \\
	&=\overbrace{\sum^{n-1}_{i = 0}(m_{x_i} m_{y_i} - m_{x_i} r_{y_i} - m_{y_i} r_{x_i})}^{P_1 \text{ and } P_2 \text{ can locally evaluate} } + \overbrace{\sum^{n-1}_{i = 0} r_{x_i} r_{y_i}  +r_z  }^{\text{Known to } P_0} \enspace . \nonumber
\end{aligned}
\end{equation}
Similar to single multiplication, $[\Gamma'] = \sum^{n-1}_{i = 0}(m_{x_i} m_{y_i} - m_{x_i} [r_{y_i}] - m_{y_i} [r_{x_i}])$ an be locally evaluated by $P_1$ and $P_2$. Meanwhile, $[\Gamma] = \sum^{n-1}_{i = 0} [r_{x_i} r_{y_i}]  +[r_z]$ can be secret shared by $P_0$ to $P_1$ and $P_2$ in the offline phase. In the online phase, $P_1$ and $P_2$ compute $[m_z] = [\Gamma] + [\Gamma']$ and reconstruct $m_z$.

\smallskip
\noindent\textbf{Security up to additive attacks.} A protocol is secure up to additive attacks when all behaviors the adversary performs can only introduce an additive error known to the adversary to the output of the protocol.
As proven in \cite{RSSadditive}, the typical replicated secret sharing protocol, such as aforementioned multiplication and inner product,  is secure up to additive attacks against malicious adversaries, i.e.,  the adversary’s cheating ability is limited to introducing an additive error to the output.

\smallskip
\noindent\textbf{Security Model.} Our protocol and framework achieve active security with abort in an honest majority setting, while one arbitrary party in $\mathcal{P}$ is under the control of a static malicious adversary. We emphasize abort security with computational soundness, ensuring that malicious behavior will be detected with overwhelming probability.

\section{3PC with Malicious Security}\label{sec:tomal}

We use the postprocessing verification procedure to detect any potential malicious behavior. Before reconstructing the final result, an extra verification is performed to ensure the correctness of the final result. Our maliciously secure protocol is based on the additive security of RSS, namely, the corresponding protocol is secure up to additive attacks. 

\noindent \textbf{Correctness Verification for Arithmetic Circuit.}
\noindent For a circuit containing both multiplication and addition gates, the correctness verification of the overall circuit using 3PC replicated shares reduces to verifying all multiplication gates. When an adversary introduces an error at an addition gate, since addition is non-interactive, it will cause an inconsistency in the shares. As previously mentioned, replicated shares possess a verifiable reconstruction property against a single malicious party. In the multiplication operation $z = x \cdot y$, $P_0$ can introduce an error when sharing $[r_x \cdot r_y]$, while $P_1$ and $P_2$ can introduce errors during the reconstruction of $m_z$, without breaking share consistency. Denoting the set of multiplication gates by $\GGG$, the verification checks the following equation:

\begin{equation}
\begin{split}
 \bigwedge_{\{x^{(i)}, y^{(i)}, z^{(i)}\} \in \GGG} x^{(i)} \cdot y^{(i)} &= z^{(i)} 
\end{split}
\label{eq: Matmul}
\end{equation}

To batch verify multiple multiplication gates ${\langle x^{(i)} \rangle, \langle y^{(i)} \rangle, \langle z^{(i)} \rangle}_{i \in |\GGG|}$, we verify that the following inner product equals zero:

\begin{equation}
\begin{split} 
\Delta = \sum^{|\GGG|}_{i = 0} ( r^i \cdot x^{(i)} \cdot y^{(i)}  - r^i \cdot z^{(i)} ) = 0
\end{split}
\label{eq: Matmul}
\end{equation}
where $r$ is a challenge picked during verification. The terms $r^i$ prevent an adversary from introducing opposing errors in different outputs $z_i$ and $z_j$ that could cancel each other. For example, if $z^{(i)} = x^{(i)} \cdot y^{(i)} + e$ and $z^{(j)} = x^{(j)} \cdot y^{(j)} - e$, then $z^{(i)} + z^{(j)} = x^{(i)} \cdot y^{(i)} + x^{(j)} \cdot y^{(j)}$, making the error undetectable.

However, directly evaluating the inner product poses challenges. One challenge is that the adversary, knowing the additive error in $\langle z^{(i)} \rangle$, could cancel out the error to fabricate $\Delta = 0$. A typical solution involves using a random factor $\alpha$. Instead of the 2-degree inner product, verification becomes a 3-degree polynomial:

\begin{equation}
\begin{split} 
\Delta = \sum^{|\GGG|}_{i = 0} ( r^i \cdot \alpha \cdot x^{(i)} \cdot y^{(i)}  - r^i \cdot \alpha \cdot z^{(i)} ) = 0
\end{split}
\label{eq: Matmul}
\end{equation}
where $\alpha$ is a random share unknown to each party. 
This randomness $\alpha$ serves as an additional layer of security by making it difficult for a malicious adversary to manipulate the values of the inputs and outputs in a way that cancels out errors introduced during verification.
If the evaluation of this 3-degree polynomial is secure against additive attacks, the adversary can only introduce an input-independent error $e'$ in $\Delta$. To cancel the original error $e$ in $z^{(i)}$, the adversary must guess $e' = \alpha \cdot e$. Since $\alpha$ is unknown and chosen randomly, the probability of correctly guessing the exact value of $\alpha \cdot e$ is extremely low.

\noindent \textbf{Ring-Specific Challenges.} The second challenge comes from irreversible multiplication in the ring.
In ring-based computations, particularly over modular arithmetic, certain errors can exploit the properties of the ring to bypass verification. For instance, an adversary could introduce a specific error $e$ such that when multiplied by $r^i$, it results in zero within the ring, even though the error itself is non-zero. Such chosen $e$ will be undetected in a high probability if a lot of values $\alpha$ meets $e \cdot \alpha = 0$.
A typical attack could involve introducing an error $e = 2^{\ell-1}$, where $\ell$ is the bit length of the ring. If $r$ is an even number, this error would result in $r^i \cdot (z^{(i)} + e) = r^i \cdot z^{(i)}$, passing verification with a probability of 1/2. 

One common solution to this problem is to increase the size of the ring used for verification, ensuring that the probability of an error passing undetected becomes vanishingly small. For example, in a protocol like SPDZ2k~\cite{spdz2k}, a larger ring size (e.g., $\ell = 100$) is used for 64-bit data, resulting in a soundness error of $2^{-36}$. In this scenario, even if the adversary tries to exploit the properties of the ring to introduce errors, the larger modulus significantly reduces the probability of success.
Since converting shares from $\ZZ_{2^{64}}$ to $\ZZ_{2^{100}}$ is expensive, it is better to perform the arithmetic directly in $\ZZ_{2^{100}}$ rather than during the verification phase, which doubles the overhead. For smaller data ranges (e.g., 1-bit values), this overhead ratio increases.

Our approach is different. We perform $\Delta$ over the extension ring $\ZZ_{2^\ell}[x]/f(x)$, where $f(x)$ is an irreducible polynomial of degree $d$ over $\ZZ_{2^\ell}$~\cite{BBC}. (The original share over $\ZZ_{2^\ell}$ becomes the free coefficient, with $d$ random elements added to the other coefficients.) According to the Schwartz-Zippel Lemma, the probability that a $|\GGG|$-degree non-zero polynomial $\Delta(r) = 0$ for a randomly chosen $r$ is at most $\frac{2^{(\ell-1)d} |\GGG| + 1}{2^{\ell d}} \approx \frac{|\GGG|}{2^d}$.

Compared to the larger ring size approach, the extension ring offers two advantages: (i) Since the share conversion to the extended ring is non-interactive, there are no modifications required during the circuit evaluation phase for the semi-honest version of the protocol. This avoids any additional communication costs typically incurred during the verification phase. (ii) The extension ring approach is compatible with the dimensionality reduction technique proposed by~\cite{GS20}, which reduces the communication complexity from $\Theta(|\GGG|)$ to $\Theta(\log |\GGG|)$. This optimization further improves the efficiency of the protocol, especially when dealing with a large number of multiplication gates.

In summary, our protocol operates as follows. First, we use a semi-honest protocol to evaluate the arithmetic circuit (on the ring $\ZZ_{2^\ell}$). We then transform all the multiplication gate triplets to the extended ring $\ZZ_{2^\ell}[x]/f(x)$ and reformulate their verification as an inner product. Next, we apply the dimension reduction method from~\cite{GS20} to reduce the $|\GGG|$-dimensional inner product to $\frac{|\GGG|}{2^R}$ dimensions. Finally, we use an inner product verification protocol to check the inner product after dimension reduction.

\myhalfbox{ Protocol $\Ptran (\{\langle x^{(i)}  \rangle,\langle y^{(i)}  \rangle, \langle z^{(i)} \rangle\}_{i \in \ZZ_N})$}{white!40}{white!10}{

    \emph{$\mathsf{Input:}$}  $N$ triples of $\langle \cdot \rangle$-shared multiplication. 
    
    \emph{$\mathsf{Output:}$} One triple of $N$-dimension $\langle \cdot \rangle^{\ell[x]}$-shared inner product. 
    
    \underline{\textbf{Preprocessing:}}
    \begin{itemize}[leftmargin=*]
    	\item[-] All parties invoke $\langle r \rangle^{\ell[x]}  \leftarrow \Pshc^{\ell[x]}$ locally;

	\end{itemize}

    \smallskip
    
    \underline{\textbf{Online:}}
    \begin{itemize}[leftmargin=*]
    \item[-] All parties reconstruct $r$ with $\Prec$ and calculate $r^i$ for all ${i \in \ZZ_N}$;
    \item[-] All parties transfer $\langle \cdot \rangle$ to $\langle \cdot \rangle^{\ell[x]}$ locally by setting the constant term of $\langle \cdot \rangle^{\ell[x]}$ to $\langle \cdot \rangle$;
    \item[-] All parties set $\langle z \rangle^{\ell[x]} := \sum^{N-1}_{i = 0} r^i \cdot \langle z^{(i)}  \rangle^{\ell[x]}$, and $\langle x'^{(i)}  \rangle^{\ell[x]} := r^i \cdot \langle x^{(i)}  \rangle^{\ell[x]}$ for all $i \in \ZZ_N$;
    \item[-] All parties output $\{\langle x'^{(i)}  \rangle^{\ell[x]},\langle y^{(i)}  \rangle^{\ell[x]}\}_{i \in \ZZ_N}; \langle z \rangle^{\ell[x]}$.
    \end{itemize}
         \medskip
    }{ Compression of Multiplication Triples.  \label{fig:MTP}}

    \smallskip
\noindent \textbf{Compression of multiplication triples.} We first design a subprotocol, $\Ptran$ (Fig.~\ref{fig:MTP}), which converts $|\GGG|$ multiplication triples over the ring $\ZZ_{2^\ell}$ into an $|\GGG|$-dimensional inner product over the polynomial ring $\ZZ_{2^\ell}[x]/f(x)$ for verification.

The transformation begins by locally converting the multiplication triples $\{\langle x^{(i)} \rangle, \langle y^{(i)} \rangle, \langle z^{(i)} \rangle\}_{i \in \ZZ_{|\GGG|}}$ to the polynomial ring equivalents $\{\langle x^{(i)} \rangle^{\ell[x]}, \langle y^{(i)} \rangle^{\ell[x]}, \langle z^{(i)} \rangle^{\ell[x]}\}{i \in \ZZ_{|\GGG|}}$. In this step, the free coefficient of the shares in $\ZZ_{2^\ell}[x]/f(x)$ is set to the original shares, while the remaining coefficients are padded with zero shares.

Next, the parties collectively generate a random challenge $r \in \ZZ_{2^\ell}[x]/f(x)$ by invoking the subprotocol $\langle r \rangle^{\ell[x]} \leftarrow \Pshc^{\ell[x]}$, followed by reconstructing $r$ via $\Prec$ (To ensure that $r$ is unknown to each party before circuit evaluation). Each party then locally computes $\langle z \rangle^{\ell[x]} = \sum_{i=0}^{|\GGG|-1} r^i \cdot \langle z^{(i)} \rangle^{\ell[x]}$ and $\langle x'^{(i)} \rangle^{\ell[x]} = r^i \cdot \langle x^{(i)} \rangle^{\ell[x]}$ for all $i \in \ZZ_{|\GGG|}$.

Finally, the protocol returns the ${|\GGG|}$-dimensional inner product tuple as $(\{\langle x'^{(i)} \rangle^{\ell[x]}, \langle y^{(i)} \rangle^{\ell[x]}\}_{i \in \ZZ_{|\GGG|}}, \langle z \rangle^{\ell[x]})$.

\begin{lemma}
\label{transl}
Suppose protocol $\Ptran$ take $\{\langle x^{(i)} \rangle,\langle y^{(i)}  \rangle, \langle z^{(i)}  \rangle\}_{i \in \ZZ_{|\GGG|}}$ as input, and it outputs 
 $\{\langle x'^{(i)} \rangle^{\ell[x]},\langle y^{(i)}  \rangle^{\ell[x]}\}_{i \in \ZZ_{|\GGG|}}; \langle z \rangle^{\ell[x]}$. The probability that the following two conditions hold is at most $\frac{|\GGG|}{2^d}$, where $d$ is the degree of $f(x)$ w.r.t. $\ZZ_{2^\ell}[x]/f(x)$:
 \begin{itemize}
\item $z = \sum^{{|\GGG|}-1}_{i = 0} x'_i \cdot y_i$
\item $\exists i \in \ZZ_{|\GGG|}$ s.t. $z_i \neq x_i \cdot y_i$
\end{itemize}
\end{lemma}

\myhalfbox{ Protocol $\Prd (\{\langle x^{(i)}  \rangle^{\ell[x]},\langle y^{(i)}  \rangle^{\ell[x]}\}_{i \in \ZZ_{|\GGG|}},\langle z \rangle^{\ell[x]})$}{white!20}{white!10}{
    
   	\emph{$\mathsf{Input:}$}  ${|\GGG|}$-dimension $\langle \cdot \rangle^{\ell[x]}$-shared inner product. 
    
    \emph{$\mathsf{Output:}$} ${|\GGG|}/2$-dimension $\langle \cdot \rangle^{\ell[x]}$-shared inner product. 
	
\underline{\textbf{Execution:}}
    \begin{itemize}[leftmargin=*]
    \item[-] For $i \in \ZZ_{{|\GGG|}/2}$, all parties set
    \begin{itemize}
    \item  $\langle f_i(0) \rangle^{\ell[x]} = \langle x^{(2\cdot i)} \rangle^{\ell[x]}$;$\langle f_i(1)\rangle^{\ell[x]} = \langle x^{(2\cdot i  + 1)} \rangle$; $\langle f_i(2)\rangle^{\ell[x]} =  2 \cdot\langle f_i(1) \rangle^{\ell[x]} - \langle f_i(0) \rangle^{\ell[x]}$;
    \item  $\langle g_i(0) \rangle^{\ell[x]} = \langle y^{(2\cdot i)} \rangle^{\ell[x]}$;$\langle g_i(1)\rangle^{\ell[x]} = \langle y^{(2\cdot i + 1)} \rangle^{\ell[x]}$; $\langle g_i(2)\rangle^{\ell[x]} =  2 \cdot\langle g_i(1) \rangle^{\ell[x]} - \langle g_i(0) \rangle^{\ell[x]}$;
    \item  $\langle h(0) \rangle^{\ell[x]} = \sum \langle f_i(0) \rangle^{\ell[x]} \cdot \langle g_i(0) \rangle^{\ell[x]}$;$\langle h(1)\rangle^{\ell[x]} = \langle z \rangle^{\ell[x]} - \langle h(0) \rangle^{\ell[x]}$; $\langle h(2)\rangle^{\ell[x]} = \sum \langle f_i(2) \rangle^{\ell[x]} \cdot \langle g_i(2) \rangle^{\ell[x]}$; 
    \end{itemize}
    	\item[-] All parties invoke $\langle \zeta \rangle^{\ell[x]} \leftarrow \Pshc^{\ell[x]}$ and reveal $\langle 2 \cdot \zeta \rangle^{\ell[x]}$;
	\item[-] All parties calculate 
	\begin{itemize}
	\item $\langle h(\zeta)\rangle^{\ell[x]} =  \sum^{2}_{i = 0}((\Pi^{2}_{j=1,j\neq i}\frac{\zeta-j}{i-j})\cdot \langle h(i)\rangle^{\ell[x]})$;
	\item $\langle f_i(\zeta)\rangle^{\ell[x]} =  \zeta \cdot\langle f_i(1) \rangle^{\ell[x]} - (\zeta - 1)\langle f_i(0) \rangle^{\ell[x]}$;
	\item $\langle g_i(\zeta)\rangle^{\ell[x]} =  \zeta \cdot\langle g_i(1) \rangle^{\ell[x]} - (\zeta - 1)\langle g_i(0) \rangle^{\ell[x]}$;

	\end{itemize}
	\item[-] All parties output  $\{\langle f_{i}(\zeta) \rangle^{\ell[x]}, \langle g_{i}(\zeta) \rangle^{\ell[x]}\}_{i \in \ZZ_{{|\GGG|}/2}}; \langle h(\zeta)\rangle^{\ell[x]}$.
    \end{itemize}

    \smallskip
}{ The Inner Product Dimension Reduction Protocol \label{fig:Reduce}} 
\begin{proof} 
It is sufficient to demonstrate that $r$ is uniformly random, assuming that the reconstruction protocol $\Prec$ does not abort. The adversary's goal is to manipulate the verification by ensuring that the following equation holds:
$$
\sum^{|\GGG|-1}_{i = 0} r^{i} \cdot z^{(i)}=\sum^{|\GGG|-1}_{i = 0} r^{i} \cdot x^{(i)} \cdot y^{(i)}
$$
where $z^{(i)} = x^{(i)} \cdot y^{(i)} + e^{(i)}$ for each $i \in \ZZ_{|\GGG|}$, and ${e_i}_{i \in \ZZ{|\GGG|}}$ represents the list of errors introduced by the adversary at each gate. 
This can be written as 
$$
\sum^{|\GGG|-1}_{i = 0} r^{i} \cdot x^{(i)} \cdot y^{(i)} = \sum^{|\GGG|-1}_{i = 0} r^{i}\cdot( x^{(i)} \cdot y^{(i)} + e^{(i)})
$$
By simplifying, we get:
$$
\sum^{|\GGG|-1}_{i = 0} r^{i} \cdot x^{(i)} \cdot y^{(i)} = \sum^{|\GGG|-1}_{i = 0} r^{i}\cdot x^{(i)} \cdot y^{(i)} + \sum^{|\GGG|-1}_{i = 0} r^{i}\cdot e^{(i)}
$$
To satisfy this equation, the adversary must ensure that the error terms cancel out, which would require:
$$
\sum^{|\GGG|-1}_{i = 0} r^{i}\cdot e^{(i)}
$$
This means that the adversary needs to find a value of $r$ that is a root of the polynomial:
$$f(x) = \sum^{|\GGG|-1}_{i = 0} x^{i} \cdot e^{(i)}$$
Since this polynomial is of degree at most ${|\GGG|} - 1$, the number of possible roots that satisfy the equation is limited. Specifically, for a degree-$\{{|\GGG|} - 1\}$ polynomial over the ring $\ZZ_{2^\ell}[x]$, according to the Schwartz-Zippel Lemma, the number of potential roots is bounded by $2^{(\ell-1)d} ({|\GGG|} + 1)$.

Thus, the probability that a uniformly random $r$ selected during the protocol coincidentally matches one of these roots is given by:
$$\frac{2^{(\ell-1)d}({|\GGG|} + 1)}{2^{\ell d}} \approx \frac{{|\GGG|}}{2^d}$$.
 \end{proof} 

 \myhalfbox{ Protocol $\Pvdot (\{\langle x^{(i)}  \rangle^{\ell[x]},\langle y^{(i)}  \rangle^{\ell[x]}\}_{i\in \ZZ_{|\GGG|}}, \langle z \rangle^{\ell[x]})$}{white!20}{white!10}{
   	\emph{$\mathsf{Input:}$}  A ${|\GGG|}$-dimension $\langle \cdot \rangle^{\ell[x]}$-shared inner product pair. 
    
    \emph{$\mathsf{Output:}$}  $z \overset{?}{=} \sum^{|\GGG|}_{i = 1} x^{(i)}  \cdot y^{(i)} $. 

\underline{\textbf{Execution:}}
    \begin{itemize}[leftmargin=*]
    \item[-] All parties invoke $\langle \alpha \rangle^{\ell[x]} \leftarrow \Pshc^{\ell[x]}$;
    \item[-] All parties calculate $\langle x'^{(i)} \rangle^{\ell[x]} = \langle x^{(i)}  \rangle^{\ell[x]} \cdot \langle \alpha  \rangle^{\ell[x]}$;
    \item[-] All parties calculate $\langle \Delta \rangle^{\ell[x]}  =  \sum^{|\GGG|}_{i = 1}\langle x'^{(i)}  \rangle^{\ell[x]} \cdot \langle y^{(i)}  \rangle^{\ell[x]} - \langle \alpha  \rangle^{\ell[x]} \cdot \langle z \rangle^{\ell[x]}$;
    \item[-] All parties call $\Delta = \Prec^{\ell[x]}(\langle \Delta \rangle^{\ell[x]})$;
    \item[-] All parties output $1$ if $\Delta = 0$, otherwise $0$.
    	    \end{itemize}

    \smallskip
}{ The Inner Product Verification Protocol \label{fig:DPV}}

\noindent \textbf{Dimension reduction.}
We extend the dimension reduction technique of  Goyal {\em et al.}~\cite{GS20} to our 3PC over ring setting. As shown in Fig.~\ref{fig:Reduce}, protocol $\Prd$ takes a shared triple $(\{\langle x^{(i)}  \rangle^{\ell[x]},\langle y^{(i)}  \rangle^{\ell[x]}\}_{i \in \ZZ_{|\GGG|}}, \langle z \rangle^{\ell[x]})$ as input  and outputs  $(\{\langle x'^{(i)}  \rangle^{\ell[x]},\langle y'^{(i)}  \rangle^{\ell[x]}\}_{i \in \ZZ_{{|\GGG|}/2}}, \langle z' \rangle^{\ell[x]})$. $\Prd$ ensures that $\sum^{{|\GGG|-1}}_{i = 0}x^{(i)}  \cdot y^{(i)}  = z$  if and only if $\sum^{{|\GGG|}/2-1}_{i = 0}x'^{(i)}  \cdot y'^{(i)}  = z'$ except for a negligible probability. At a high level, for the inner product input $\{x^{(i)} \}_{i \in \ZZ_{|\GGG|}}$ and $\{y^{(i)} \}_{i \in \ZZ_{|\GGG|}}$, we can utilize $x^{(2i)}$ and $x^{(2i-1)}$ to interpolate ${|\GGG|}/2$ linear functions $\{f_{i}(\cdot)\}_{i \in \ZZ_{{|\GGG|}/2}}$ at the point $0$ and $1$, and similarly interpolate  $\{g_{i}(\cdot)\}_{i \in \ZZ_{{|\GGG|}/2}}$ by $\{y^{(i)} \}_{i \in \ZZ_{|\GGG|}}$. Considering the correct output $z$, we have $$z = \sum^{{|\GGG|}/2}_{i = 0} f_{i}(0)\cdot g_{i}(0) + f_{i}(1)\cdot g_{i}(1)$$
Let $h(\cdot) = \sum^{{|\GGG|}/2}_{i = 0}  f_{i}(\cdot) \cdot g_{i}(\cdot)$. This leads to the equation $h(1) = z - h(0)$.
The protocol $\Prd$ computes $h(0) = \sum^{{|\GGG|}/2}_{i = 0}  f_{i}(0) \cdot g_{i}(0)$ and $h(2)=\sum^{{|\GGG|}/2}_{i = 0}  f_{i}(2) \cdot g_{i}(2)$, and from this, it calculates $h(1) = z - h(0)$. Then, $\Prd$ interpolates the polynomial $h(x)$ using the values $h(0)$, $h(1)$, and $h(2)$. Finally, all parties choose a random point $\zeta$ and output the new shared triple $(\{\langle f_{i}(\zeta) \rangle^{\ell[x]}, \langle g_{i}(\zeta) \rangle^{\ell[x]}\}_{i \in \ZZ_{{|\GGG|}/2}}, \langle h(\zeta)\rangle^{\ell[x]})$, which preserves the inner product relation if and only if the initial condition $z = \sum^{{|\GGG|}/2}_{i = 1} f_{i}(0)\cdot g_{i}(0) + f_{i}(1)\cdot g_{i}(1)$ holds.

It is important to note that the points 0, 1, and 2 correspond to ring elements with free coefficients of 0, 1, and 2 in $\ZZ_{2^\ell}[x]/f(x)$.

The protocol $\Prd$ requires one round of communication involving $5 \ell \cdot d$ bits in the online phase and one round involving $\ell \cdot d$ bits in the offline phase. We execute $\Prd$ $R$ times to reduce the inner product dimension to ${|\GGG|} / 2^R$, after which the resulting vectors are verified by checking


$$\sum^{{|\GGG|} / 2^R}_{i = 0}\langle f_i(\zeta) \rangle^{\ell[x]} \cdot \langle g_i(\zeta) \rangle^{\ell[x]}=\langle h(\zeta) \rangle^{\ell[x]}$$ We prove the soundness error of the $\Prd$ is $\frac{1}{2^{d-1}}$ in Lemma~\ref{polyred}.

\myhalfbox{ Protocol $\Pmulv^R (\{\langle x^{(i)} \rangle,\langle y^{(i)} \rangle, \langle z^{(i)} \rangle\}_{i \in \ZZ_{|\GGG|}})$}{white!20}{white!10}{
    \emph{$\mathsf{Input:}$}  ${|\GGG|}$ pairs of $\langle \cdot \rangle$-shared multiplication. 
    
    \emph{$\mathsf{Output:}$}  $z^{(i)} \overset{?}{=}  x^{(i)} \cdot y^{(i)}$ for all $i \in \ZZ_{|\GGG|}$. 
    
   \underline{\textbf{Execution:}}	
    \begin{itemize}[leftmargin=*]
    \item[-] All parties invoke $\Ptran(\{\langle x^{(i)} \rangle,\langle y^{(i)} \rangle; \langle z^{(i)} \rangle\}_{i \in \ZZ_{|\GGG|}})$ to get $\{\langle x^{(i)} \rangle^{\ell[x]},\langle y^{(i)} \rangle^{\ell[x]}\}_{i \in \ZZ_{|\GGG|}};\langle z \rangle^{\ell[x]}$;
    \item[-] For $k = 1,\ldots, R$, all parties perform:
       \begin{itemize}
   	\item $\{\{\langle x^{(i)} \rangle^{\ell[x]},\langle y^{(i)} \rangle^{\ell[x]}\}_{i \in \ZZ_{{|\GGG|}/2^k}};\langle z \rangle^{\ell[x]} \} \leftarrow \Prd (\{\langle x^{(i)} \rangle^{\ell[x]},\langle y^{(i)} \rangle^{\ell[x]}\}_{i \in \ZZ_{{|\GGG|}/2^{k-1}}};\langle z \rangle^{\ell[x]})$;
	\end{itemize}
	\item[-] All parties invoke $b = \Pvdot (\{\langle x^{(i)} \rangle^{\ell[x]},\langle y^{(i)} \rangle^{\ell[x]}\}_{i \in \ZZ_{{|\GGG|}/2^R}};\langle z \rangle^{\ell[x]})$;
	\item[-] All parties output $b$.
    \end{itemize}

    \smallskip
}{ The Batch Multiplication Verification Protocol \label{fig:PPEV}}

\begin{lemma}
\label{polyred}
Suppose $\Prd$ take $(\{\langle x^{(i)}  \rangle^{\ell[x]},\langle y^{(i)}  \rangle^{\ell[x]}\}_{i \in \ZZ_{|\GGG|}},\langle z \rangle^{\ell[x]})$ as input, and it outputs the new list 
$(\{\langle x'^{(i)}  \rangle^{\ell[x]}, \langle y'^{(i)}  \rangle^{\ell[x]}\}_{i \in \ZZ_{{|\GGG|}/2}}, \langle z'\rangle^{\ell[x]})$. The probability that the following two conditions hold is at most $\frac{1}{2^{d-1}}$, where $d$ is the degree of $f(x)$ w.r.t. $\ZZ_{2^\ell}[x]/f(x)$:
 \begin{itemize}
\item $z' = \sum^{{|\GGG|}/2}_{i = 0} x'^{(i)}  \cdot y'^{(i)} $
\item $z \neq \sum^{{|\GGG|}}_{i = 0} x^{(i)}  \cdot y^{(i)} $
\end{itemize}
\end{lemma}

\begin{proof} 
For clarity, we denote \( h'(k) = \sum^{{|\GGG|}/2}_{i = 0} f_i(k) \cdot g_i(k) \). The adversary’s goal is to manipulate the computation such that \( h(\zeta) = h'(\zeta) \), while also ensuring that 
\[
h(0) + h(1) = h'(0) + h'(1) + e,
\]
where \( e \) represents the error introduced in \( z \). Simultaneously, the adversary can introduce new errors \( e_1 \) and \( e_2 \) during the calculation of \( h(0) \) and \( h(2) \), such that:

\[
h(0) = h'(0) + e_1, \quad h(1) = h'(1) + e - e_1, \quad h(2) = h'(2) + e_2.
\]
Considering the Lagrange interpolation for randomly chosen \( \zeta \in \ZZ_{2^\ell}[x] \), we have:
\begin{equation*}
\begin{split}
 h(\zeta) &= \sum^{2}_{i = 0} \left( \prod^{2}_{\substack{j=0 \\ j \neq i}} \frac{\zeta - j}{i - j} \right) \cdot h(i) = \frac{(\zeta - 1)(\zeta - 2)}{2} \cdot h(0)\\
  &+ \zeta(2 - \zeta) \cdot h(1) + \frac{(\zeta - 1)\zeta}{2} \cdot h(2)
\end{split}
\label{eq:2}
\end{equation*}
and for \( h'(\zeta) \):

\[
h'(\zeta) = \frac{(\zeta - 1)(\zeta - 2)}{2} \cdot h'(0) + \zeta(2 - \zeta) \cdot h'(1) + \frac{(\zeta - 1)\zeta}{2} \cdot h'(2).
\]
To ensure \( h(\zeta) = h'(\zeta) \), the adversary must satisfy the following equation:

\[
\frac{(\zeta - 1)(\zeta - 2)}{2} \cdot e_1 + \zeta(2 - \zeta) \cdot (e - e_1) + \frac{(\zeta - 1)\zeta}{2} \cdot e_2 = 0
\]

The probability that the adversary can choose \( e, e_1, e_2 \) such that this equation holds is equivalent to making \( \zeta \) a root of the degree-2 polynomial:

\[
f(x) = \frac{(x - 1)(x - 2)}{2} \cdot e_1 + x(2 - x) \cdot (e - e_1) + \frac{(x - 1)x}{2} \cdot e_2
\]
over \( \ZZ_{2^\ell}[x] \), which has at most \( 2^{2(\ell - 1)d} + 1 \) roots. Therefore, the soundness error is:

\[
\frac{2^{(\ell - 1)d + 1} + 1}{2^{\ell d}} \approx \frac{1}{2^{d-1}}.
\]

 \end{proof} 

 \myhalfbox{ Protocol $\Pmul (\langle x \rangle, \langle y \rangle)$}{white!40}{white!10}{
    \emph{$\mathsf{Input:}$} $\langle \cdot \rangle$-shared value $x,y$. 
  	
  \emph{$\mathsf{Output:}$} $\langle \cdot \rangle$-shared value $z$ where $z = x \cdot y$.
  
    \underline{\textbf{Preprocessing:}}
    \begin{itemize}[leftmargin=*]
    \item[-] All parties prepare $[r_z] \leftarrow \Psha$ locally;
	\item[-] $P_0$ calculates $\Gamma = r_{x}\cdot r_{y} + r_z$ and shares it with $\Psha(\Gamma)$;
	
	\end{itemize}

    \smallskip
    
    \underline{\textbf{Online:}}
    \begin{itemize}[leftmargin=*]
    \item[-]$P_j$ for $j \in \{1,2\}$ calculates $[m_z]_j =(j-1)m_{x} \cdot m_{y} - m_{x} [r_{y}]_j - m_{y_i} [r_{x}]_j + [\Gamma]$ and mutually exchange their shares to reconstruct $m_z$.
    \end{itemize}
         \medskip
     \underline{\textbf{Postprocessing:}}
     \begin{itemize}[leftmargin=*]
    \item[-]For all multiple gate wire value $\{\langle x^{(i)} \rangle, \langle y^{(i)} \rangle, \langle z^{(i)} \rangle\}_{i \in \ZZ_{|\GGG|}}$, all parties call $\Pmulv^R (\{\langle x^{(i)} \rangle, \langle y^{(i)} \rangle; \langle z^{(i)} \rangle\}_{i \in \ZZ_{|\GGG|}})$ to verify correctness.
    \end{itemize}
    }{ The Multiplication Protocol   \label{fig:MultP}}

\noindent \textbf{Inner product verification.}
Our inner product verification protocol, denoted as $\Pvdot$ (Fig.~\ref{fig:DPV}), verifies the inner product relationship of shared values over the polynomial ring $\ZZ_{2^\ell}[x]/f(x)$. Specifically, to verify the relation 
\[
\sum_{i=0}^{|\GGG| / 2^R} \langle x^{(i)} \rangle^{\ell[x]} \cdot \langle y^{(i)} \rangle^{\ell[x]} = \langle z \rangle^{\ell[x]},
\]
$\Pvdot$ checks whether the expression
\[
\langle \alpha \rangle^{\ell[x]} \cdot (\sum_{i=0}^{|\GGG| / 2^R} \langle x^{(i)} \rangle^{\ell[x]} \cdot \langle y^{(i)} \rangle^{\ell[x]} - \langle z \rangle^{\ell[x]})
\]
is equal to zero.

Unfortunately, as far as we know, there is currently no semi-honest 3PC protocol that securely evaluates a cubic (degree-3) polynomial while being resilient to additive attacks. As an alternative, we compute $x'^{(i)} = \alpha \cdot x^{(i)}$ for each $i \in \ZZ_{|\GGG|}$. Subsequently, all parties evaluate the inner product 
\[
\sum_{i=0}^{|\GGG| / 2^R} x'^{(i)} \cdot y^{(i)}.
\]

This method, however, does not achieve complete security against additive attacks, as an adversary may introduce an error $e'^{(i)}$ into $x'^{(i)}$, resulting in an overall error term:
\[
\sum_{i=0}^{|\GGG| / 2^R} e'^{(i)} \cdot y^{(i)},
\]
which is dependent on $y^{(i)}$. Nevertheless, considering that $y^{(i)}$ is obtained via multiple Lagrange interpolations in the prior dimension reduction protocol, $y^{(i)}$ can be treated as a random value.

Let $e$ denote the error in $z$. The adversary must guess $\alpha \cdot e + \sum_{i=0}^{|\GGG| / 2^R} e'^{(i)} \cdot y^{(i)} = 0$, where $y^{(i)}$ is effectively random. The probability of success for this guess is $\frac{1}{2^d}$.

We prove in Lemma~\ref{polyvv} that the soundness error of the $\Pvdot$ protocol is $\frac{1}{2^d}$.
\myhalfbox{ Protocol $\Pdot (\langle x_1 \rangle,\ldots,\langle x_{n} \rangle, \langle y_1 \rangle,\ldots,\langle y_{n} \rangle)$}{white!40}{white!10}{
    \emph{$\mathsf{Input:}$} $\langle \cdot \rangle$-shared value list of $x_i$ and $y_i$. 
    
    \emph{$\mathsf{Output:}$} $\langle \cdot \rangle$-shared value of $z$ where $z = \sum^n_{i = 1} x_i \cdot y_i$. 
    
    \underline{\textbf{Preprocessing:}}
    \begin{itemize}[leftmargin=*]
    	\item[-] All parties prepare $[r_z] \leftarrow \Psha$ locally;
	\item[-] $P_0$ calculates $\Gamma = \sum^{n}_{i = 1} r_{x_i}\cdot r_{y_i} + r_z$ and shares it with $\Psha(\Gamma)$;

	\end{itemize}

    \smallskip
    
    \underline{\textbf{Online:}}
    \begin{itemize}[leftmargin=*]
    \item[-]$P_j$ for $j \in \{1,2\}$ calculates $[m_z]_j = \sum^{n}_{i=1} (j-1)m_{x_i} \cdot m_{y_i} - m_{x_i} [r_{y_i}]_j - m_{y_i} [r_{x_i}]_j + [\Gamma]_j$ and mutually exchange their shares to reconstruct $m_z$.
    \end{itemize}
    
    \underline{\textbf{Postprocessing:}}
    \begin{itemize}[leftmargin=*]
    \item[-]For ${|\GGG|}$ pairs inner product result $\{\{\langle x^{(j)}_i \rangle,\langle y^{(j)}_i \rangle\}_{i \in \ZZ_{n_j}}; \langle z^{(j)} \rangle\}_{j \in \ZZ_{|\GGG|}}$, all parties call $\Pvdot^R (\{\{\langle x^{(j)}_i \rangle,\langle y^{(j)}_i \rangle\}_{i \in \ZZ_{n_j}}; \langle z^{(j)} \rangle\}_{j \in \ZZ_{|\GGG|}})$ to verify correctness.
    \end{itemize}
         \medskip
    }{ The Inner Product Protocol   \label{fig:DPP}}  

\begin{lemma}
\label{polyvv}
Let $(\{\langle x^{(i)} \rangle^{\ell[x]},\langle y^{(i)} \rangle^{\ell[x]}\}_{i\in \ZZ_{|\GGG|}}, \langle z \rangle^{\ell[x]})$ be the input of  protocol $\Pvdot$ depicted in Fig.~\ref{fig:DPV}. The probability that $\Pvdot$ outputs $1$ and $z \neq \sum^{{|\GGG|}-1}_{i = 0} x^{(i)} \cdot y^{(i)}$  is at most $\frac{1}{2^d}$, where $d$ is the degree of $f(x)$ w.r.t. $\ZZ_{2^\ell}[x]/f(x)$.
\end{lemma}
\begin{proof} 
 
Since $\alpha$ is uniformly random and unknown to the adversary, for $z = \sum_{i=0}^{|\GGG|} x^{(i)} \cdot y^{(i)} + e$, we have 
\[
\Delta = \alpha \cdot e + \sum_{i=0}^{|\GGG| / 2^R} e'^{(i)} \cdot y^{(i)},
\]
where $e'^{(i)}$ is introduced during the evaluation of $\alpha \cdot x^{(i)}$. Given that 3PC multiplication is secure up to additive attacks, $e'^{(i)}$ is independent of $\alpha$. Therefore, we can treat $\sum_{i=0}^{|\GGG| / 2^R} e'^{(i)} \cdot y^{(i)}$ as an overall error term $e'$.

By the Schwartz-Zippel Lemma, the polynomial $f(x) = e \cdot x + e'$ over the ring $\ZZ_{2^\ell}[x]$ has at most $2^{(\ell - 1)d} + 1$ roots. Consequently, the probability that the adversary can deliberately choose $e$ such that $\Delta = 0$ is 
\[
\frac{2^{(\ell - 1)d} + 1}{2^\ell d} \approx \frac{1}{2^d}.
\]

 \end{proof} 



Our batch multiplication verification protocol $\Pmulv$ in Fig.~\ref{fig:PPEV} integrates the above three subroutines, which requires one round communication of $(R + {|\GGG|}/2^R)\ell \cdot d $ bits in the offline phase and  $R+2$-round  communication of $(5R + 3 + {|\GGG|}/2^R)\ell \cdot d $ bits in the online phase for ${|\GGG|}$ multiplication triples. We prove soundness error of $\Pmulv$ is  $\frac{{|\GGG|}}{2^{d - R - 2}}$ in Thm.~\ref{multvv}.  

\begin{theorem}
\label{multvv}
Let $\{\langle x^{(i)}  \rangle,\langle y^{(i)}  \rangle, \langle z^{(i)}  \rangle\}_{i \in \ZZ_{|\GGG|}}$ be the input of protocol $\Pmulv^R$ depicted in Fig.~\ref{fig:PPEV}. The probability 
$\Pmulv^R$ outputs $1$ and $\exists i \in \ZZ_{|\GGG|}$ s.t. $z^{(i)}  \neq x^{(i)}  \cdot y^{(i)} $ is at most  $\frac{{|\GGG|}}{2^{d - R - 2}}$, where $d$ is the degree of $f(x)$ w.r.t. $\ZZ_{2^\ell}[x]/f(x)$.
\end{theorem}
\begin{proof} 
 
From Lemma~\ref{transl}, Lemma~\ref{polyred}, and Lemma~\ref{polyvv}, we know that the adversary has $R$ chances with success probability $\frac{1}{2^{d-1}}$, one chance with probability $\frac{|\GGG|}{2^d}$, and one chance with probability $\frac{1}{2^d}$ to pass the verification.

Therefore, the total probability that the adversary succeeds is 
\[
1 - \left( 1 - \frac{1}{2^{d-1}} \right)^R \cdot \left( 1 - \frac{|\GGG|}{2^d} \right) \cdot \left( 1 - \frac{1}{2^d} \right) \approx \frac{|\GGG|}{2^{d - R - 2}}.
\]

 \end{proof}

    \myhalfbox{ Protocol $\Pbsv^R (\{\{\langle x^{(j)}_i \rangle,\langle y^{(j)}_i \rangle\}_{i \in \ZZ_{n_j}}, \langle z^{(j)} \rangle\}_{j \in \ZZ_{|\GGG|}})$}{white!20}{white!10}{
   	\emph{$\mathsf{Input:}$} ${|\GGG|}$ pairs of inner product. 
    
    \emph{$\mathsf{Output:}$} Output if  $z^{(j)} = \sum^n_{i = 1} x^{(j)}_i \cdot y^{(j)}_i$ held for all $j \in \ZZ_{|\GGG|}$. 

\underline{\textbf{Execution:}}
    \begin{itemize}[leftmargin=*]
    \item[-] All parties transfer all shares $\langle \cdot \rangle$ to $\langle \cdot \rangle^{\ell[x]}$ locally;
    \item[-] All parties invoke $\langle r \rangle^{\ell[x]} \leftarrow \Pshc^{\ell[x]}$ an call $\Prec$ to reconstruct $r \in \ZZ_{2^\ell}[x]$;
    \item[-] All parties set $\langle z \rangle^{\ell[x]} := \sum r^j \cdot \langle z^{(j)} \rangle^{\ell[x]}$ and $\langle x^{(j)}_i \rangle^{\ell[x]} := r^j \cdot \langle x^{(j)}_i \rangle^{\ell[x]}$ for each $i \in \ZZ_{n_j}, j \in \ZZ_{|\GGG|}$;
    \item[-]  All parties consolidate the original pairs into a single pair $\{\langle x^{(i)} \rangle^{\ell[x]},\langle y^{(i)} \rangle^{\ell[x]}\}_{i\in \ZZ_\mathcal{{|\GGG|}}}; \langle z \rangle^{\ell[x]}$ where $\mathcal{{|\GGG|}} = \sum^{{|\GGG|}-1}_{j = 0} n_j$;
    \item[-] For $k = 1,\ldots, R$, all parties do:
       \begin{itemize}
   	\item $\{\langle x^{(i)} \rangle^{\ell[x]},\langle y^{(i)} \rangle^{\ell[x]}\}_{i \in \ZZ_{\mathcal{{|\GGG|}}/{2^k}}},\langle z \rangle^{\ell[x]} \leftarrow \Prd (\{\langle x_i \rangle^{\ell[x]},\langle y^{(i)} \rangle^{\ell[x]}\}_{i \in \ZZ_{\mathcal{{|\GGG|}}/2^{k-1}}}, \langle z \rangle^{\ell[x]})$;
	\end{itemize}
	\item[-] All parties call $b = \Pvdot (\{\langle x^{(i)} \rangle^{\ell[x]},\langle y^{(i)} \rangle^{\ell[x]}\}_{i \in \ZZ_{\mathcal{{|\GGG|}}/{2^R}}}, \langle z \rangle^{\ell[x]})$;
	\item[-] All parties output $b$.
    \end{itemize}

    \smallskip
}{ The Batch Inner Product Verification Protocol \label{fig:PBDP}}   
\section{Enhancing PPML.}\label{apparithmetic} 
In this section, we implement a maliciously secure privacy-preserving machine learning framework. We use boolean share to evaluate nonlinear functions, which can be viewed as share over ring $\ZZ_2$. We realize the share conversion protocol, which is entirely based on maliciously secure multiplication $\Pmul$. This makes our framework merely reliant on $\Pmul$.
\subsection{Dealing with linear operation.} Our maliciously secure multiplication protocol is shown in Fig.~\ref{fig:MultP}. $\Pmul$ ensures the correctness of multiplication by invoking batch verification protocol $\Pmulv$ in the post-processing phase. When handling a substantial volume of data, our protocol exhibits an amortized communication of $\ell$ bits in the preprocessing phase and $2 \ell$ bits in the online phase for each multiplication operation. The multiplication protocol can be expanded to the inner product protocol.  Our maliciously secure inner product protocol $\Pdot$ is shown in Fig.~\ref{fig:DPP}. Its semi-honest version is the special case of $\Pple$ for $2$-degree $n$-variate polynomial, which requires one round communication of $\ell$ bits in the preprocessing phase and one round communication of $2\ell$ bits in the online phase. To extend it to the malicious setting, we employ batch verification protocol $\Pvdot^R$ (Fig.~\ref{fig:PBDP}) to ensure the correctness of the inner products with a similar manner of multiplication. Analogously, in $\Pvdot^R$, all parties transform the verification of inner product triples over ring $\ZZ_{2^\ell}$ to the verification of a single inner product triple over the polynomial ring $\ZZ_{2^\ell}[x]/f(x)$. Following that, all parties invoke $\Prd$ to reduce the dimension of the vector that needs to be verified. 
When handling a substantial volume of data, on average, our protocol exhibits an amortized communication of $\ell$ bits in the preprocessing phase and $2 \ell$ bits in the online phase for each inner product operation.
In the application of machine learning, we view the $m$-dimensional output convolution and matrix multiplication as $m$ separate inner products. We implement these two types of operations by invoking $\Pdot$ a total of $m$ times.  

\myhalfbox{ Protocol $\Ptrunc^t(\langle x \rangle)$}{white!40}{white!10}{
    Let $\mathsf{rshift}(x, y)$ denote right shift $x$ with $y$ bits.
    $\mathsf{Input:}$ $\langle \cdot \rangle$-shared value. 
    
    $\mathsf{Output:}$ $\langle \cdot \rangle$-shared value of $z = \mathsf{rshift}(x, t)$. 
    
    \underline{\textbf{Preprocessing:}}
    \begin{itemize}[leftmargin=*]
    \item[-] $P_0$ and $P_i$ pick random bit list $\{b_{i,j}\}_{j \in Z_\ell} \leftarrow \ZZ^\ell_2$ together, for $i\in \{1,2\}$;
    \item[-] All parties set 
    \begin{itemize}
    \item[-] $\langle b_{1,j} \rangle := (m_{b_{1,j}}, [r_{b_{1,j}}]_1, [r_{b_{1,j}}]_2) := (0, b_{1,j}, 0)$;
    \item[-] $\langle b_{2,j} \rangle:= (m_{b_{2,j}}, [r_{b_{2,j}}]_1, [r_{b_{2,j}}]_2) := (0, 0, b_{2,j})$ for $j \in \ZZ_{\ell}$;
    \end{itemize}
    \item[-] All parties invoke $\Pdot$ to calculate
    \begin{itemize}
    \item$\langle r_x\rangle = \sum^{\ell -1 }_{j = 0} 2^j (\langle b_{1,j} \rangle + \langle b_{2,j} \rangle - 2\langle b_{1,j} \rangle \cdot \langle b_{2,j} \rangle)$;
    \item$\langle r_z\rangle= \sum^{\ell - t -1 }_{j = 0} 2^j (\langle b_{1, j+t} \rangle + \langle b_{2, j+t} \rangle - 2\langle b_{1, j+t} \rangle \cdot \langle b_{2, j+t} \rangle)+\sum^{\ell - 1 }_{j = \ell - t -1} 2^j (\langle b_{1, \ell - 1} \rangle + \langle b_{2, \ell - 1} \rangle - 2\langle b_{1, \ell - 1} \rangle \cdot \langle b_{2, \ell - 1} \rangle)$;
    
   	\end{itemize}
	\item[-] $P_0$ set $r_x = \sum^{\ell -1 }_{j = 0} 2^j \cdot (b_{1,j} \oplus b_{2,j})$, $r_z= \sum^{\ell - t -1 }_{j = 0} 2^j \cdot (b_{1,j} \oplus b_{2,j}) +\sum^{\ell - 1 }_{j = \ell - t -1} 2^j \cdot (b_{1,\ell - 1} \oplus b_{2,\ell - 1})$;
	
	\item[-] $P_i$ for $i \in \{1,2\}$ set $[r_x] = m_{r_x}  - [r_{r_x}] $, $[r_z] = m_{r_z}  - [r_{r_z}] $;
	\end{itemize}
	\smallskip
	\underline{\textbf{Online:}}
	
	 \begin{itemize}[leftmargin=*]        
	 \item[-] $P_i$ for $i \in \{1,2\}$ set $m_z = \mathsf{rshift}(m_x, t)$;
	 \item[-] All parties output $\langle z \rangle := ([r_z], m_z)$.
        \end{itemize}

    }{ The maliciously secure truncation protocol \label{fig:trunc}}
\subsection{Secure Truncation Protocol.} The multiplication of two fixed-point values with our encoding will lead to a double scale of $2^k$ for the fractional precision $k$. An array of protocols \cite{aby3, BLAZE,swift} using the probabilistic truncation protocol to reduce the additional $2^k$ scaler. Their protocols introduce a one-bit error which is caused by the carry bit of truncated data. In addition, the probabilistic truncation protocol makes an error with a certain probability (assuming that the valid range of data is $\ell_x$ and the error probability is $2^{\ell_x - \ell + 1}$). As shown in Fig.~\ref{fig:trunc}, we also design a maliciously secure probabilistic truncation protocol $\Ptrunc^t$ for the truncation bit size $t$. Our idea is similar to SWIFT \cite{swift}, which generates correct truncation pair via maliciously secure inner product protocol. However, in contrast to SWIFT\cite{swift}, we directly generate $r_z = \mathsf{rshift}(r_x, d)$, which allows the parties locally truncate $m_z = \mathsf{rshift}(m_x, d)$ in the online phase without communication. Although SWIFT\cite{swift} eliminates communication by combining truncation with multiplication, they still need $2\ell$ online communication in the online phase of the standalone truncation protocol.
\myhalfbox{ Protocol $\Pdabits$}{white!40}{white!10}{
    
    $\mathsf{Input:}$ None. 
    
    $\mathsf{Output:}$ edaBits pair $\{\langle r \rangle^\ell, \{\langle r[i] \rangle^1\}_{i \in [\ell]}\}$. 
    
    \underline{\textbf{Execution:}}
    \begin{itemize}[leftmargin=*]
    \item[-] $P_0$ and $P_j$ pick random bit list $\{b_{i,j}\}_{i \in [\ell]} \leftarrow (\ZZ_2)^\ell$ together, for $i\in \{1,2\}$;
    \item[-] $P_1$ and $P_2$ pick random bit list $\{m_{i}\}_{i \in [\ell]} \leftarrow (\ZZ_2)^\ell$ together, for $i\in \{1,2\}$;
    \item[-] All parties set $\langle r[i] \rangle^1 := (m_i, b_{i, 1}, b_{i, 2})$;
    \item[-] All parties set 
    \begin{itemize}
    \item[-] $\langle b_{i,1} \rangle^\ell := (0, b_{i,1}, 0)$;
    \item[-] $\langle b_{i,2} \rangle^\ell:= (0, 0, b_{i,2})$;
    \item[-] $\langle m_i \rangle^\ell:= (m_i, 0, 0)$ for $i \in [\ell]$;
    \end{itemize}
    \item[-] All parties invoke $\Pmul$ to calculate
    \begin{itemize}
    \item $\langle r'[i]\rangle^\ell =  \langle b_{i,1} \rangle^\ell + \langle b_{i,2} \rangle^\ell - 2\langle b_{i,1} \rangle^\ell \cdot \langle b_{i,2} \rangle^\ell$ for $i \in [\ell]$;
    \item $\langle r\rangle^\ell =  \sum^{\ell - 1}_{i = 0} 2^i \cdot (\langle m_i \rangle^\ell + \langle r'[i] \rangle^\ell - 2\langle m_i \rangle^\ell \cdot \langle r'[i] \rangle^\ell)$
    
   	\end{itemize}
	\item[-] All parties output $\{\langle r \rangle^\ell, \{\langle r[i] \rangle^1\}_{i \in [\ell]}\}$
	\end{itemize}

    }{ The maliciously edaBits generation\label{fig:edabits}}
Specifically, we let $P_0$ and $P_1$ pick random bit list $\{b_{1,j}\}_{j \in Z_\ell}$ together; $P_0$ and $P_2$ pick random bit list $\{b_{2,j}\}_{j \in Z_\ell}$ together. We utilize these lists to calculate that $r_x = \sum^{\ell -1 }_{j = 0} 2^j \cdot (b_{1,j} \oplus b_{2,j})$ and $r_z= \sum^{\ell - t -1 }_{j = 0} 2^j \cdot (b_{1,j} \oplus b_{2,j}) +\sum^{\ell - 1 }_{j = \ell - t -1} 2^j \cdot (b_{1,\ell - 1} \oplus b_{2,\ell - 1})$ which keeps the relationship $r_z = \mathsf{shift}(r_x, t)$. We can evaluate $r_x$ and $r_z$ under $\langle \cdot \rangle$-sharing to realize malicious security. To transform $b_{1,j}$ and $b_{2,j}$ to the $\langle \cdot \rangle$-sharing locally, we let $\langle b_{1,j} \rangle = (0, b_{1,j}, 0)$ and $\langle b_{2,j} \rangle = (0, 0, b_{2,j})$ which set the other secret share to be $0$. For the result $\langle r_x \rangle$ and $\langle r_z \rangle$, since $r_x$ and $r_z$ is known by $P_0$, $P_1$ and $P_2$ can be locally calculate $[r_x] = m_{r_x}  - [r_{r_x}] $ and $[r_z] = m_{r_z}  - [r_{r_z}]$. Note that $\Ptrunc$ requires assigning $r_x$ of the input wire, we let it be executed preferentially to provide $r_x$ for the other gate.
Our maliciously secure protocol $\Ptrunc$ requires $1$ rounds and communication of $6\ell$ bits in the offline phase and requires no communication in the online phase.

\subsection{Secure Non-linear Function Evaluation.} We utilize standard daBits~\cite{edabits, dabits} to evaluate non-linear functions. DaBits convert arithmetic secret sharing into Boolean secret sharing, allowing us to use Boolean shares for circuit evaluation corresponding to the non-linear function. In our protocol, a daBit consists of a pair ${\langle r \rangle^\ell, \langle r \rangle^1}$, where $r \in {0,1}$. As a multi-bit version, edaBits~\cite{edabits} is a pair $\{\langle r \rangle^\ell, {\langle r[i] \rangle^1}_{i \in [\ell]}\}$, where $r \in \mathbb{Z}{2^\ell}$ is an arithmetic share, and ${\langle r[i] \rangle^1}_{i \in [\ell]}$ are Boolean shares resulting from bit extraction. We demonstrate that our maliciously secure 3PC protocol is fully compatible with edaBits.

EdaBits allow bit extraction from shares over $\mathbb{Z}_{2^\ell}$ into multiple shares over $\mathbb{Z}_2$. A series of works~\cite{aby3} use full adders to perform this conversion. Specifically, let $\langle x \rangle^\ell$ be the share requiring bit extraction and $\{\langle r \rangle^\ell, \{\langle r[i] \rangle^1\}_{i \in [\ell]}\}$ be the edaBits. All parties first reconstruct $\Delta = x - r$ by locally computing $\langle \Delta \rangle^\ell = \langle x \rangle^\ell - \langle r \rangle^\ell$. Subsequently, using $\{\langle r[i] \rangle^1\}_{i \in [\ell]}$ and $\{\langle \Delta[i] \rangle^1\}_{i \in [\ell]}$, the parties perform a full-adder circuit to obtain $\{\langle x[i] \rangle^1\}_{i \in [\ell]}$. Using $x[i]$, each party can then evaluate arbitrary non-linear functions. Note that the soundness of our malicious 3PC protocol is independent of the ring size $\ell$, and it works for $\ell = 1$ as well.

For the share conversion from $\langle x \rangle^1$ to $\langle x \rangle^\ell$, this can be achieved through multiplication. Given the edaBits $\{\langle r \rangle^\ell, \langle r \rangle^1\}$, all parties first reconstruct $\Delta = \langle x \rangle^1 - \langle r \rangle^1$ over $\mathbb{Z}_2$. Then, they compute $\langle x \rangle^\ell = \Delta + \langle r \rangle^\ell - 2\Delta \cdot \langle r \rangle^\ell$. This works because $\Delta \oplus r = \Delta + r - 2 \Delta \cdot r$.

The edaBits $\{\langle r \rangle^\ell, \{\langle r[i] \rangle^1\}_{i \in [\ell]}\}$ can be constructed from multiple daBits $\{\langle r[i] \rangle^\ell, \{\langle r[i] \rangle^1\}_{i \in [\ell]}\}$ by computing $\langle r \rangle^\ell = \sum_{i=0}^{\ell-1} 2^i \cdot \langle r[i] \rangle^\ell$.

To generate daBits $\{\langle r \rangle^\ell, \langle r \rangle^1\}$, parties $P_0$ and $P_1$ jointly pick a random value $r_1 \in \{0,1\}$, $P_0$ and $P_2$ pick $r_2 \in \{0,1\}$, and $P_1$ and $P_2$ pick $r_3 \in \{0,1\}$. All parties then set $\langle r_1 \rangle^\ell := (0, r_1, 0)$, $\langle r_2 \rangle^\ell := (0, 0, r_2)$, and $\langle r_3 \rangle^\ell := (r_3, 0, 0)$. The parties set $\langle r \rangle^1 := (r_3, r_1, r_2)$, and compute $\langle r' \rangle^\ell = \langle r_1 \rangle^\ell + \langle r_2 \rangle^\ell - 2\langle r_1 \rangle^\ell \cdot \langle r_2 \rangle^\ell$. Finally, they compute $\langle r \rangle^\ell = \langle r' \rangle^\ell + \langle r_3 \rangle^\ell - 2\langle r' \rangle^\ell \cdot \langle r_3 \rangle^\ell$. 

Note that $\langle r \rangle^\ell$ requires two rounds of multiplication. Considering edaBits $\{\langle r \rangle^\ell, \{\langle r[i] \rangle^1\}_{i \in [\ell]}\}$, the second multiplication can be combined with a single inner product. Fig.~\ref{fig:edabits} depicts the generation of edaBits. 

Our PPML framework is constructed by multiplication over $\ZZ_{2^\ell}$ and $\ZZ_2$, which can be verified by perform $\Pmulv$ on $\ZZ_{2^\ell}$ and $\ZZ_2$ respectively.

\begin{figure}[tbp]%
  \centering
\begin{subfigure}[t]{.48\linewidth}
\includegraphics[width=1\textwidth]{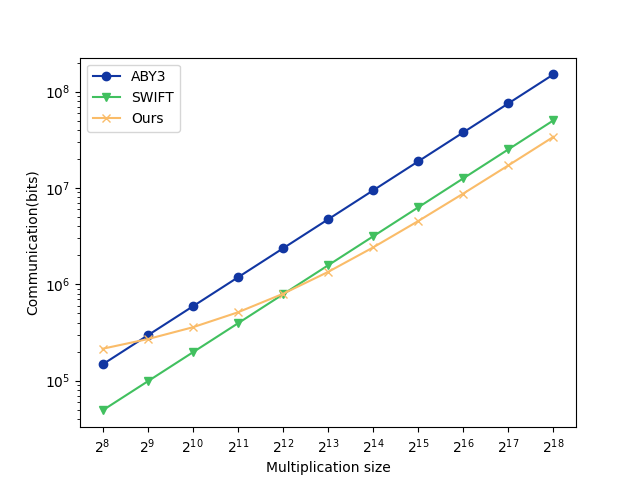}
\caption{Online com. of MUL}
\end{subfigure}
\begin{subfigure}[t]{.48\linewidth}
\includegraphics[width=1\textwidth]{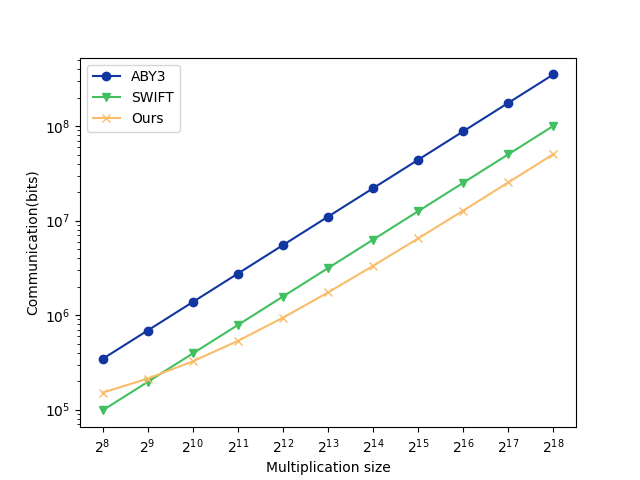}
  \caption{Total com. of MUL}
 \end{subfigure}
\begin{subfigure}[t]{.48\linewidth}
\includegraphics[width=1\textwidth]{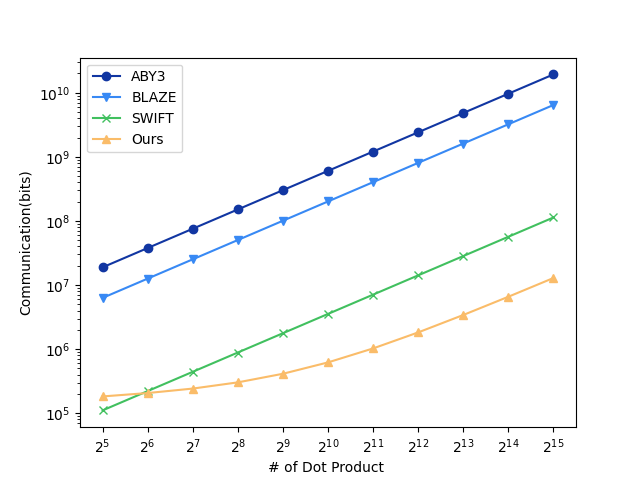}
\caption{Online communication of Inner Product with Trunction}
\end{subfigure}
\begin{subfigure}[t]{.48\linewidth}
\includegraphics[width=1\textwidth]{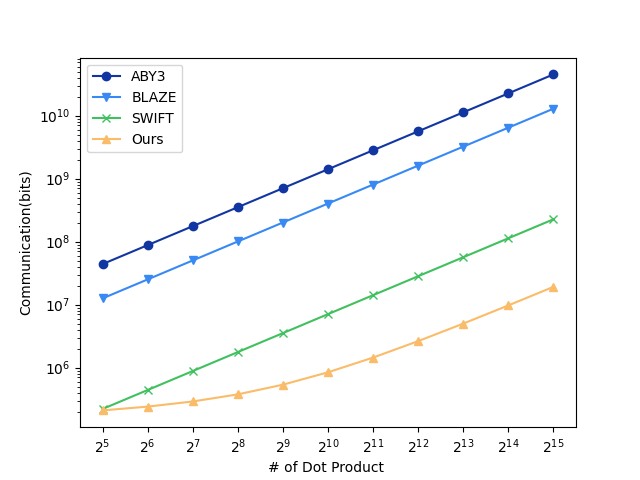}
  \caption{Overall communication of Inner Product with Trunction}
 \end{subfigure}
 
  \caption{ Communication overhead comparison with ABY3\cite{aby3}, BLAZE\cite{BLAZE}, SWIFT\cite{swift} of muliplication and inner product.}%
  \label{fig:com_dot}%
  \vspace{-1em}
\end{figure}

\section{Implementation and Benchmarks}\label{imp_ben}

\noindent In this section, we evaluate our multiplication and non-linear protocols in both the semi-honest and malicious settings. For the maliciously secure multiplication protocols, we compare the communication and runtime with SWIFT\cite{swift} and ABY\cite{aby3}. 

\smallskip
\noindent \textbf{Benchmark setting.}  We perform our arithmetic protocols on the GPU setting. To support GPU, our code is based on the Piranha~\cite{piranha} source code \cite{Piranhacode}, which is a GPU platform for MPC protocols. For the non-linear protocols, we implement both CPU and GPU versions to support benchmarking with  FSS~\cite{DCF} and garble circuit-based protocol BLAZE~\cite{BLAZE} on CPU setting.   The modified version of Piranha (GPU-version)~\cite{modified}, and the CPU version of our code~\cite{ourcode} are available in Anonymize Github.
In our benchmark setting, we take the size of the ring $\ell = 64$ and the polynomial ring degree $d = 64$. For the fixed-point value, we utilize $16$ bits truncation. Our experiments are performed in a local area network, using software to simulate three network settings: local-area network (LAN, RTT: 0.2ms, bandwidth: 1Gbps), metropolitan-area network (MAN, RTT: 12ms, bandwidth: 100Mbps), and wide-area network (WAN, RTT: 80ms, bandwidth: 40Mbps) and executed on a desktop with AMD Ryzen 7 5700X CPU @ 3.4 GHz running Ubuntu 18.04.2 LTS; with 8 CPUs, 32 GB Memory, $4 \times$ Nvidia 2080 Ti with 11 GB RAM and 1TB SSD.

\begin{figure}[tpb]%
  \centering

\begin{subfigure}[t]{.49\linewidth}
\includegraphics[width=1.05\textwidth]{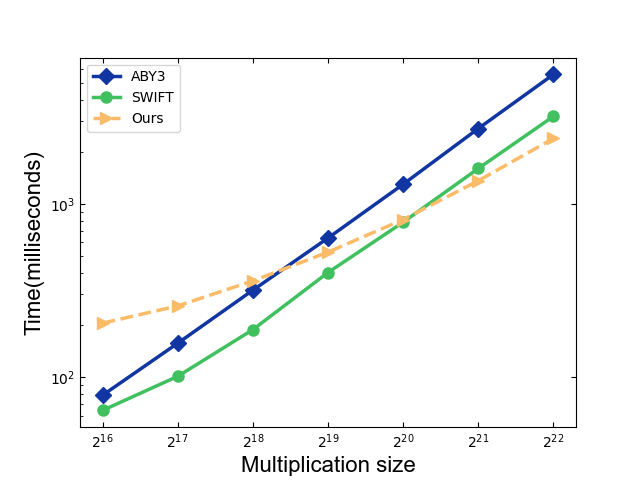}
  \caption{MAN}
 \end{subfigure}
\begin{subfigure}[t]{.49\linewidth}
\includegraphics[width=1.05\textwidth]{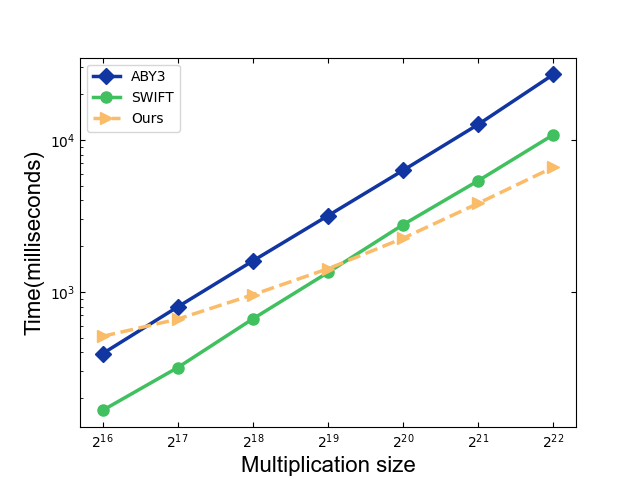}
\caption{WAN}
 \end{subfigure}

  \caption{ Overall running time of multiplication (over the GPU setting). Compared with ABY3\cite{aby3},  SWIFT\cite{swift} of $\Pmul$ over MAN and WAN setting.}%
  
  \label{fig:multrun}%

\vspace{-1em}
\end{figure}

\begin{figure}[tbp]%
  \centering
\begin{subfigure}[t]{.49\linewidth}
\includegraphics[width=1.05\textwidth]{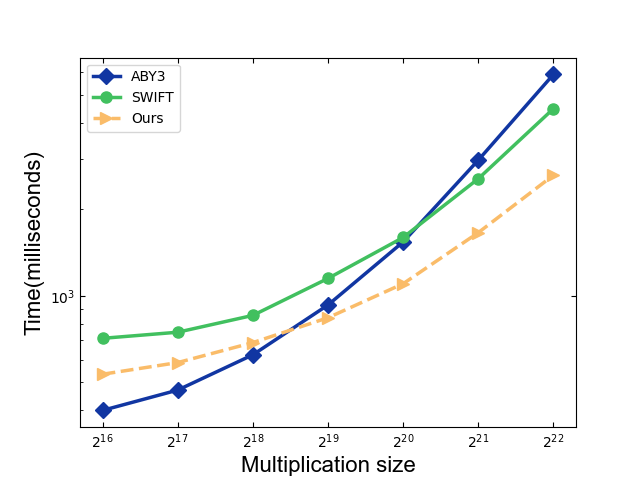}
\caption{32 Depth}
\end{subfigure}
\begin{subfigure}[t]{.49\linewidth}
\includegraphics[width=1.05\textwidth]{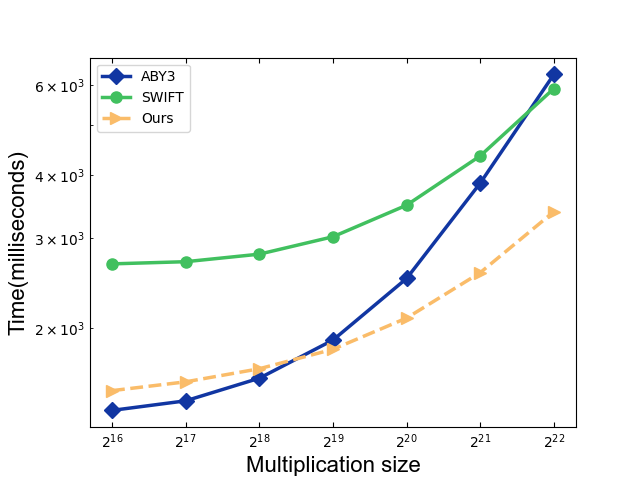}
  \caption{128 Depth}
 \end{subfigure}

  \caption{ Evaluate the multiplication (over the GPU setting) with circuit depth $32$ and $128$ under the MAN setting.}%
  \label{fig:dep}%
  \vspace{-1em}
\end{figure}

\subsection{Multiplication performance comparison}\label{sec:mulcomm}
In this section, we benchmark our maliciously secure 3PC of multiplication.

\smallskip
\noindent \emph{\underline{Trade-off of the repetition parameter $R$.}} While selecting a larger value for the repetition parameter $R$ for dimension reduction can minimize the communication volume in batch verification, it is also essential to consider the impact of additional communication rounds in the postprocessing phase for overall performance.  We conduct a practical experimental benchmark to determine the optimal value of $R$ in different bandwidth and delay scenarios. 
Fig.~\ref{fig:trade-off1} depicts the verification time with the different dimension reduction number $R$. The running time is measured in milliseconds on the y-axis, while the x-axis shows the number of dimension reductions $R$, ranging from 0 to 10. It points out the optimal $R$ value ($R=7$ in MAN, with data size $2^{18}$; $R=9$ in MAN, with data size $2^{20}$; $R=8$ in WAN, with data size $2^{18}$;$R=10$ in WAN, with data size $2^{20}$;). Our benchmark indicates that the larger $R$ needs to be chosen for smaller bandwidths and larger data dimensions.

\smallskip
\noindent\textbf{Communication.}
 Fig.~\ref{fig:com_dot} presents the communication overhead of our protocol compared to ABY, BLAZE, and SWIFT. For this evaluation, we consider a vector dimension of $1024$ when computing the inner product. Our protocol incurs a logarithmic additional communication cost of $(6R + 5)\ell \cdot d$, where $R = \log |\GGG|$. As a result, when $|\GGG|$ is small, our protocol requires more communication than SWIFT.
However, as $|\GGG|$ increases, the logarithmic term $R$ diminishes in relative significance, making the additional overhead negligible for large $|\GGG|$. In scenarios with large input sizes, the communication volume of our protocol for multiplication is approximately $50\%$ that of SWIFT and $15\%$ that of ABY. For the inner product computation in 1024 dimensions with truncation, the communication overhead is about $50\%$ that of SWIFT and $0.02\%$ that of ABY.
This demonstrates that while our protocol may initially have higher communication costs for smaller input sizes, it becomes significantly more efficient as the input size increases. Specifically, it provides notable savings in communication compared to both SWIFT and ABY in large-scale computations, particularly in higher-dimensional inner product operations.

\smallskip
\noindent\textbf{Performance.}
We compare our maliciously secure multiplication protocol with state-of-the-art (SOTA) protocols in Fig.~\ref{fig:multrun}, including SWIFT and ABY. To assess the running time, we execute our protocol across multiple values of $R$ (the dimension reduction factor), selecting the configuration that yields the best performance. Due to the inclusion of an additional verification round—especially in scenarios with a small volume of data—our protocol experiences a performance overhead, making it slower than SWIFT and ABY in these cases. This verification round is the dominant source of overhead when the data size is small.
However, when considering larger, saturated datasets, our protocol demonstrates a significant performance improvement, achieving up to $2\times$ the throughput of both SWIFT and ABY under both MAN (Metropolitan Area Network) and WAN (Wide Area Network) settings.

Furthermore, we investigate the effect of multiplication depth on protocol performance. Fig.~\ref{fig:dep} illustrates how performance varies with different multiplication depths. We benchmark our protocol against others on multiplication circuits with depths of 32 and 128. Due to the ability of our protocol and ABY to leverage batch verification, they exhibit a round complexity advantage, resulting in better performance compared to the SWIFT protocol, especially when the multiplication depth is large. This round reduction becomes particularly beneficial as circuit depth increases, allowing our protocol to outperform SWIFT in deep circuit scenarios.

In summary, while our protocol may incur some overhead in low-volume, low-depth scenarios due to the additional verification step, it scales efficiently with larger datasets and deeper circuits, making it highly competitive against SOTA protocols such as SWIFT and ABY in high-volume, high-depth applications.

\begin{figure}[tpb]%
  \centering

\begin{subfigure}[t]{.49\linewidth}
\includegraphics[width=1.05\textwidth]{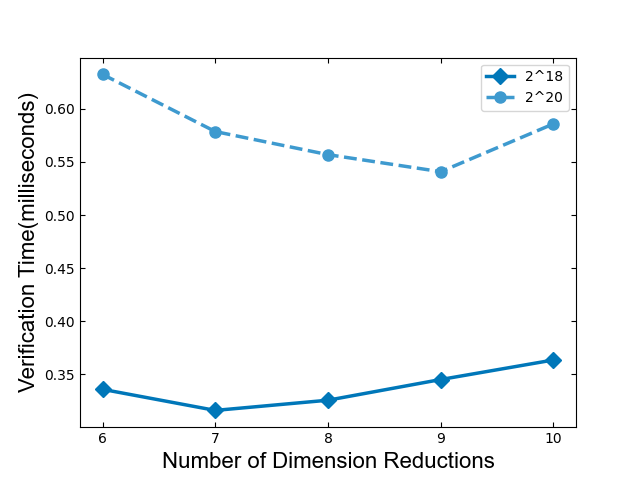}
  \caption{MAN}
 \end{subfigure}
\begin{subfigure}[t]{.49\linewidth}
\includegraphics[width=1.05\textwidth]{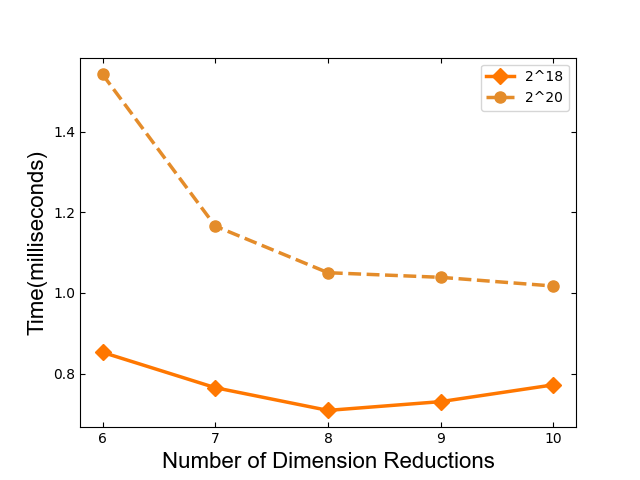}
\caption{WAN}
 \end{subfigure}

  \caption{ The running time of verification phase (over the GPU setting), with the different dimension reduction number $R$, multiplication triple size $2^{18}$ and $2^{20}$, over MAN and WAN setting.}%
  
  \label{fig:trade-off1}%

\end{figure}
\begin{figure*}
\includegraphics[width=1.05\textwidth]{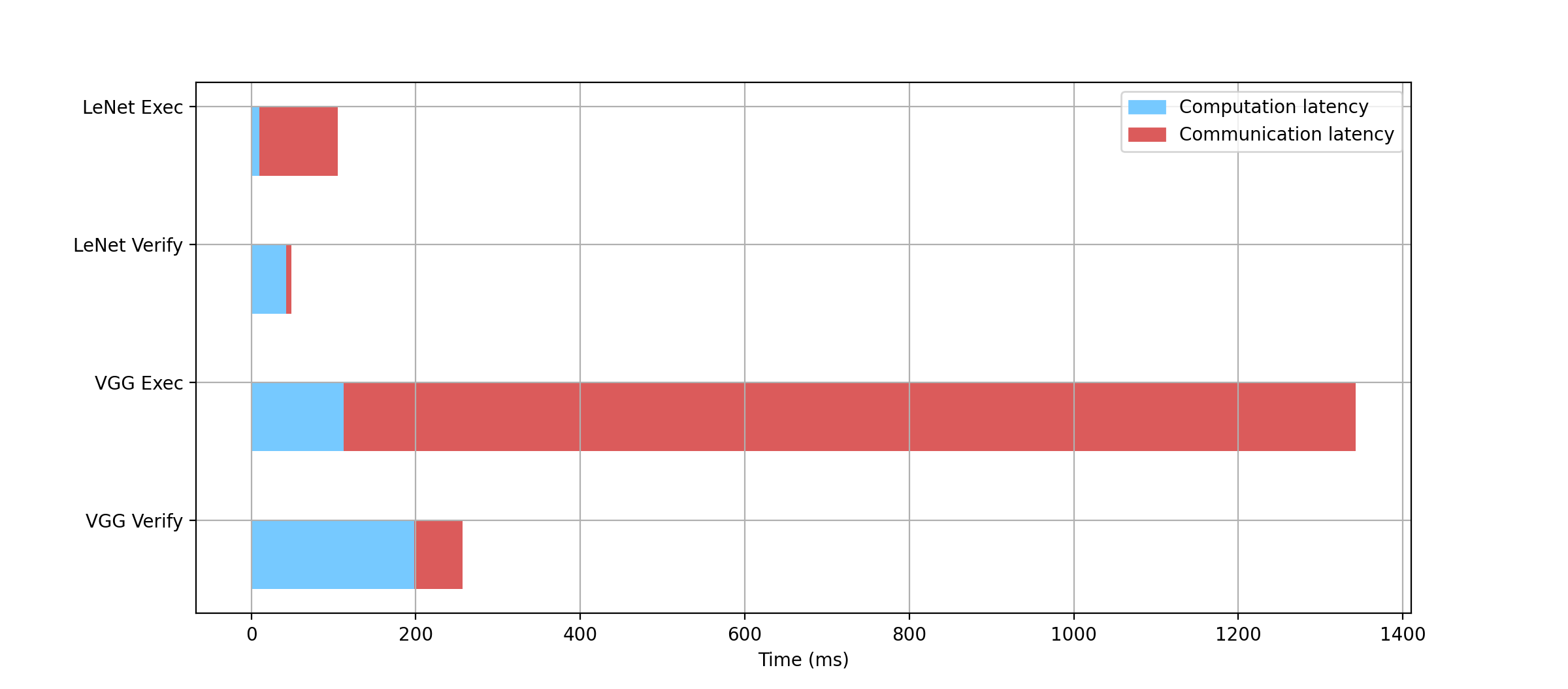}
\caption{Computation latency and communication latency diagram for model inference in LAN setting. Exec corresponds to the execution phase (both offline and online), and Verify corresponds to the verification phase. }
\label{fig:latency}
 \end{figure*}
\subsection{The inference of neural network.}\label{overall}
\begin{table}
\caption{Run-time and communication cost of NN inference, under LAN setting with batch size $30$. (Com: the communication which is given in MB. Time: the run-time which is given in ms)}

\label{table:nn}
\centering

\resizebox{1\columnwidth}{!}{
\begin{tabular}{c|c|cc|ccc}
\toprule
\multirow{2}{*}{Model} & \multirow{2}{*}{Stage} & \multicolumn{2}{c|}{Offline}    & \multicolumn{3}{c}{Online}                                   \\ \cmidrule(r){3-7} 
                       &                           & \multicolumn{1}{c|}{Com} & Time & \multicolumn{1}{c|}{Com} & \multicolumn{1}{c|}{Round} & Time\\ \midrule
\multirow{2}{*}{S-NN}      &       Execution          & \multicolumn{1}{c|}{0.05}    & 6.07 & \multicolumn{1}{c|}{0.17}    & \multicolumn{1}{c|}{2}      &  13.19 \\ \cmidrule(r){2-7} 
                       &       Verification       & \multicolumn{1}{c|}{-}    &   -   & \multicolumn{1}{c|}{1.75}    & \multicolumn{1}{c|}{3}      &   23.52\\ \midrule
\multirow{2}{*}{LeNet}      &     Execution         & \multicolumn{1}{c|}{0.65}    &  7.40  & \multicolumn{1}{c|}{2.46}    & \multicolumn{1}{c|}{42}      &   104.9  \\ \cmidrule(r){2-7} 
                       &         Verification         & \multicolumn{1}{c|}{-}    &   -   & \multicolumn{1}{c|}{26.1}    & \multicolumn{1}{c|}{3}      &  48.2 \\ \midrule
\multirow{2}{*}{VGG}      &   Execution      & \multicolumn{1}{c|}{10.2}    &207& \multicolumn{1}{c|}{39.2}    & \multicolumn{1}{c|}{127}      &  1341    \\ \cmidrule(r){2-7} 
                       &Verification & \multicolumn{1}{c|}{-}    &   -   & \multicolumn{1}{c|}{414}    & \multicolumn{1}{c|}{6}      &     257 \\ \bottomrule

\end{tabular}}

\end{table}
We further construct the convolutional neural network (CNN) inference. We implement three types of models as follows:

\begin{itemize}

\item Shallow neural network(S-NN). Our shallow neural network accepts $28 \times 28$ image and involves a convolution layer(5 kernels with $5 \times 5$ shape, the stride of (2,2)), a ReLU layer, and a fully connected layer(connects the incoming $5 \times 13 \times 13$ nodes to the output $10$ nodes).
\item LeNet. We benchmark the LeNet model, which replaces the sigmoid activation layer with the ReLU layer. The model accepts $32 \times 32$ image and contains 2-layer convolution, 2-layer Maxpool, 4-layer ReLU, and 3-layer full connection.
\item VGG-16. We benchmark the VGG-16 model, which takes $64 \times 64$ image as input and contains 13-layer convolution, 5-layer maxpool, 13-layer ReLU, and 8-layer full connection.

\end{itemize}
TABLE~\ref{table:nn} presents the performance metrics of our protocol across different stages for three models: S-NN, LeNet, and VGG. The metrics are divided into communication (Com) and time (Time).  During the execution stage, all parties engage in both the offline and online phases of the semi-honest protocol, ensuring the necessary computations are securely executed. 
The communication costs rise significantly with model complexity, especially in the execution phase. 
In the verification stage, a post-processing step is performed by all parties to validate the correctness of the computed and shared results.
Our platform demonstrates the capability to execute CNNs, such as LeNet, in mere hundreds of milliseconds. For deeper convolutional networks, like VGG, the platform completes the execution within seconds, showcasing its efficiency even with more complex neural network models. As observed in the provided table, when the model is small, such as in the case of the S-NN model, the verification phase takes a significant proportion of the total runtime compared to the online execution phase. However, as the model complexity increases, such as with LeNet and VGG, the runtime required for the online execution phase grows substantially while the proportion of the verification phase decreases. In summary, for small models, verification is a major contributor to total runtime. However, as the model size increases, the overhead of verification decreases in proportion and becomes less impactful, with the online phase becoming the dominant factor in overall runtime.

Fig.~\ref{fig:latency}  illustrates the computation and communication latencies during the model inference process. The figure details the latencies for two models, LeNet and VGG, across the execution phase (Exec) and the verification phase (Verify).
The results highlight that, during the execution phase, communication latency far exceeds computation latency. This effect is particularly pronounced for LeNet, where communication latency dominates the entire inference process, underscoring the significant impact of network communication on overall performance. In contrast, during the verification phase, computation latency becomes the primary factor contributing to the total latency.
As the complexity of the model increases, as seen when comparing VGG to LeNet, the communication overhead becomes the main bottleneck in the online phase, severely impacting inference efficiency. Meanwhile, in the verification phase, the computation overhead takes precedence, becoming the dominant factor that limits overall performance. These findings emphasize the need for optimizing both communication and computation aspects, particularly as model complexity increases, to improve inference efficiency across all stages.

\section{related work}

\cite{FastRing} achieves a communication overhead of $1 \frac{1}{3}$ ring elements with two rounds of communication or $1 \frac{2}{3}$ ring elements with one round of communication.
With the advancement of the maliciously secure multiplication protocol, practical maliciously secure privacy-preserving machine learning becomes attainable. 
\cite{SecureML,ASTRA,aby3, ASTRA, FALCON, Trident, FLASH, BLAZE, swift} realize privacy-preserving machine learning protocols under the malicious threat model in an honest majority.
In the semi-honest setting, protocols such as \cite{aby3, ASTRA, aby2, BLAZE} are all based on three parties replicated secret sharing, which only request $3$ ring elements communication each multiplication. 
The online phase communication overhead of $2$ ring elements can be achieved by handing over part of the communication to a circuit-dependent offline phase~\cite{ASTRA}.
For malicious security, 3PC and 4PC protocols with an honest majority, such as ASTRA~\cite{ASTRA}, SWIFT~\cite{swift}, and others~\cite{Fantastic, Tetrad, aby3, BLAZE} have been developed for private machine learning and advanced applications. These protocols combine function-dependent preprocessing and mixed-protocol strategies, providing malicious security while optimizing performance. 
In particular, a series of optimizations \cite{ASTRA, BLAZE, swift} reduced the multiplication overhead to 6 ring elements ($3$ in the offline phase) in the three-party setting.
Other notable contributions include efforts to reduce the cost of garbled circuits against a malicious evaluator~\cite{Pragmatic} and improving communication efficiency in settings with a dishonest majority and semi-honest helper party~\cite{Asterisk}. Additionally, protocols like MUSE~\cite{Muse} and follow-up works~\cite{SIMC,Fusion} explored fixed-corruption scenarios in 2PC, enhancing malicious security while reducing performance overhead.

\section{Conclusion}
In this work, we design a 3PC maliciously secure multiplication protocol over the ring. We adopt our protocol for the machine learning model evaluation and perform the model evaluation on the GPU platform. The experiments show that our various protocols have significant performance improvements over the state-of-the-art works, i.e., \cite{swift, aby3}.



\bibliographystyle{IEEEtran}
\bibliography{ref}

%
%
%
%
%
%
%
%
%
%
\vspace{22pt}

\end{document}